\documentclass[11pt]{article}
\usepackage{amssymb}
\usepackage{amsthm}
\usepackage{amsmath}
\usepackage{geometry}
\usepackage{enumitem}
\usepackage{tikz}
\usepackage{url}
\usepackage{graphicx}
\usepackage{xcolor}
\usepackage{pgfplots}
\usepgfplotslibrary{fillbetween}
\usetikzlibrary{shapes.geometric}
\usetikzlibrary{arrows}

\pgfplotsset{ticks=none,compat=newest}

\numberwithin{equation}{section}

\newtheorem{thm}{Theorem}[section]
\newtheorem{prop}[thm]{Proposition}
\newtheorem{lem}[thm]{Lemma}
\newtheorem{cor}[thm]{Corollary}

\theoremstyle{definition}
\newtheorem{defn}[thm]{Definition}

\theoremstyle{remark}
\newtheorem*{rem}{Remark}
\newtheorem*{exmp}{Example}

\newcommand{\diff}{\mathop{}\!\mathrm{d}}
\newcommand{\supp}{\mathop{}\!\mathrm{supp}}
\newcommand{\Aut}{\mathop{}\!\mathrm{Aut}}
\newcommand{\scp}[2]{\langle #1, #2\rangle}
\newcommand{\diracscp}[2]{\langle #1| #2\rangle}
\newcommand{\e}{\mathop{}\!\mathrm{e}}
\newcommand{\I}{\mathop{}\!\mathrm{i}}
\newcommand{\outgoing}{\mathop{}\!\stackrel{\mathrm{out}}{\times}}
\newcommand{\norm}[1]{\left\lVert#1\right\rVert}
\newcommand{\ArSp}[1]{\sigma_\alpha ({#1})}
\newcommand{\jap}[1]{\langle #1 \rangle}
\newcommand{\normm}[1]{{\vert\kern-0.25ex\vert\kern-0.25ex\vert #1\vert\kern-0.25ex\vert\kern-0.25ex\vert}}
\renewcommand{\div}{\nabla\cdot}

\title{Mourre theory and asymptotic observables \\ in local relativistic quantum field theory}
\author{Janik Kruse \\\\\
	\small{Adam Mickiewicz University in Pozna\'n} \\
	\small{Faculty of Mathematics and Computer Science} \\  	
	\small{ul.~Uniwersytetu Pozna\'nskiego 4, 61-614 Pozna\'n, Poland} \\
	\small{E-mail: \tt{janik.kruse@amu.edu.pl}}}
\date{September 16, 2024}

\begin{document}

\maketitle

\begin{abstract}
	We prove the convergence of Araki--Haag detectors in any Haag--Kastler quantum field theory with an upper and lower mass gap. We cover the case of a single Araki--Haag detector on states of bounded energy, which are selected from the absolutely continuous part of the energy-momentum spectrum sufficiently close to the lower boundary of the multi-particle spectrum. These states essentially encompass those states in the multi-particle spectrum lying below the three-particle threshold. In our proof, we draw on insights from proofs of asymptotic completeness in quantum mechanics. Notably, we apply Mourre's conjugate operator method for the first time within the framework of Haag--Kastler quantum field theory. Furthermore, we discuss applications of our findings for the problem of asymptotic completeness in local relativistic quantum field theory.
\end{abstract}

\noindent Keywords: Haag--Kastler quantum field theory, Haag--Ruelle scattering theory, asymptotic completeness, asymptotic observables, Araki--Haag detectors, Mourre's conjugate operator method

\vspace*{0.1cm}

\noindent Mathematics Subject Classification 2020: 81Q10, 81T05, 81U99

\section{Introduction} \label{sec:Introduction}

A fundamental task of scattering theory is to prove asymptotic completeness, which is important for interpreting quantum theories in terms of particles. In non-relativistic quantum mechanics, asymptotic completeness for $N$-particle Hamiltonians has been established through the works of Enss~\cite{enss1984}, Sigal and Soffer~\cite{sigal1987}, Graf~\cite{graf1990}, Yafaev \cite{yafaev1993}, Dereziński~\cite{derezinski1993}, and many others (see~\cite{derezinski1997} for a textbook exposition). These classical results rely on the existence of asymptotic observables such as the asymptotic velocity.\footnote{The cited papers are formulated mainly in time-dependent scattering theory. More recently, Skibsted~\cite{skibsted2023} provided a proof of asymptotic completeness for short-range interactions in time-independent scattering theory.}

In local relativistic quantum field theory, however, asymptotic completeness remains an open problem, even at the level of two particles in massive theories. Asymptotic completeness has been proved only for few models, including integrable models \cite{lechner2007} and the $P(\phi)_2$ model at the level of two \cite{spencer1976} and three particles \cite{combescure1982}. 

The reason for this gap lies in additional conceptual and technical difficulties in quantum field theory, as discussed in \cite{dybalski2014_2}. Firstly, within the conventional Haag--Kastler framework, there are pathological counterexamples to asymptotic completeness (e.g.~generalised free fields). Moreover, the infinite number of degrees of freedom in quantum field theory allows for a rich superselection structure and myriads of elusive charged particles. As a result, the vacuum sector accommodates not only neutral particles but also particle states containing oppositely charged pairs of particles. On the technical side, the understanding of dynamical properties of systems with relativistic dispersion relation is incomplete. Specifically, for $N\geq3$, asymptotic completeness is an open problem for $N$-particle Hamiltonians with non-quadratic dispersion relation.

It is not surprising that experts in quantum field theory have focused on other properties than asymptotic completeness. A closely related problem is the convergence of asymptotic observables corresponding to the asymptotic velocity in quantum mechanics. The convergence of asymptotic observables is still a difficult problem in Haag--Kastler quantum field theory, where observables possibly evolve through a cascade of charged particles and pathological states. However, building on decades of previous research outlined below, we make substantial progress on this question.  

Araki--Haag detectors have long ago been identified as natural asymptotic observables in quantum field theory. In their seminal paper, Araki and Haag \cite{araki1967} proved the convergence of these asymptotic observables on incoming scattering states ($\mathcal{H}^\mathrm{in}$) and outgoing scattering states ($\mathcal{H}^\mathrm{out}$) of bounded energy and interpreted them as particle counters.\footnote{Notably, Enss~\cite{enss1975} published a paper on Araki--Haag detectors before his celebrated proof of asymptotic completeness in three-particle quantum mechanical systems. Asymptotic observables play a central role in Enss' proof of asymptotic completeness.} However, the convergence of Araki--Haag detectors on arbitrary states of bounded energy has remained an open problem for decades despite continued interest related to various aspects of particles in quantum field theory \cite{buchholz1986, buchholz1995, porrmann2004, buchholz2006} \cite[Section~VI.1]{haag1996}.

First convergence results of Araki--Haag detectors on arbitrary states of bounded energy have been obtained relatively recently by Dybalski and Gérard \cite{dybalski2014, dybalski2014_2}. They managed to translate quantum mechanical methods such as large-velocity and phase-space propagation estimates to Haag--Kastler quantum field theory via a technically important uniform bound by Buchholz~\cite{buchholz1990}. Dybalski and Gérard covered products of two or more Araki--Haag detectors sensitive to particles with distinct velocities, but products of detectors sensitive to particles with coinciding velocities and the case of a single detector were not treated. The technical reason for this omission was a missing low-velocity propagation estimate, which is usually proved by Mourre's conjugate operator method \cite[Theorem~4.13.1]{derezinski1997} \cite{hunziker1999,richard2004}. 

Mourre's method is a powerful mathematical technique from spectral theory, which is based on a strictly positive commutator estimate. In the appendix, we provide a concise overview of some important results of this method. The conjugate operator method led to significant progress in the spectral and scattering theory of many-body Schrödinger operators, but it resisted so far any extension from quantum mechanics to Haag--Kastler quantum field theory. Through scattering theory, we manage to apply Mourre's method to quantum field theory. This allows us to prove the convergence of a single Araki--Haag detector on states of bounded energy sufficiently close to the lower boundary of the multi-particle spectrum.

To state our main result and explain the essence of our argument, we introduce some notation. Let $h\in C_c^\infty(\mathbb{R}^s)$ be a smooth compactly supported function and $B^*$ a creation operator (i.e.~a bounded operator that creates one-particle states; see Section \ref{sec:HRScatteringTheory} for precise definitions). Araki--Haag detectors are asymptotic limits ($t\to\pm\infty$) of the following observable: 
\begin{align} \label{eq:Detector}
	C(h,t) = \e^{\I t H} \int_{\mathbb{R}^s} h\left(\frac{\mathbf{x}}{t}\right) (B^*B)(\mathbf{x}) \diff \mathbf{x} \e^{-\I t H}.
\end{align}
It is possible to extend the above formula to $h\in L^\infty(\mathbb{R}^s)$ by Buchholz's uniform estimate \cite{buchholz1990}. Let $\mathcal{H}_{\mathrm{ac}}(P)$ be the jointly absolutely continuous spectral subspace of the energy-momentum operator $P=(H,\mathbf{P})$. We denote the spectral measure of $P$ by $E$. Our main result is that $C(h,t)$ converges strongly on states $\psi\in\mathcal{H}_{\mathrm{ac}}(P)$ for which $B\psi$ is a one-particle state. We formulate the result for the limit $t\to\infty$ and outgoing scattering states. The result for the limit $t\to-\infty$ is analogous if $\mathcal{H}^\mathrm{out}$ is replaced by $\mathcal{H}^\mathrm{in}$.

\begin{thm}\label{thm:MainResult}
	Let $\Delta\subset\mathbb{R}^d$ be compact and $\psi\in E(\Delta)\mathcal{H} \cap \mathcal{H}_{\mathrm{ac}}(P)$ a state of bounded energy. If $B^*$ is a creation operator such that $B\psi$ is a one-particle state, then, for all $h\in L^\infty(\mathbb{R}^s)$, $C(h,t)\psi$ converges strongly in $\mathcal{H}$ as $t\to\infty$. If $\psi$ lies in the orthogonal complement of the scattering states $\mathcal{H}^{\mathrm{out}}$, then the limit is 0.
\end{thm}

Intuitively, the condition that $B\psi$ is a one-particle state selects states $\psi$ of the multi-particle spectrum below the three-particle threshold (see~Figure~\ref{fig:Spectrum} and the comments preceding Theorem~\ref{thm:NonDetectionNonScatteringStates} for more details). We emphasise that the assumptions of the theorem exclude neither a non-trivial superselection structure nor pathological states with too many degrees of freedom.

Regarding the spectral assumption $\psi\in \mathcal{H}_{\mathrm{ac}}(P)$, we note that the Hilbert space $\mathcal{H}$ decomposes into the pure point, absolutely continuous, and singular continuous spectral subspace of $P$:
\begin{align}
	\mathcal{H}=\mathcal{H}_{\mathrm{pp}}(P) \oplus \mathcal{H}_{\mathrm{ac}}(P) \oplus \mathcal{H}_{\mathrm{sc}}(P).
\end{align}
Typically, the pure point spectral subspace $\mathcal{H}_{\mathrm{pp}}(P)$ is the span of the vacuum vector $\Omega$, and the singular continuous spectral subspace $\mathcal{H}_{\mathrm{sc}}(P)$ describes mass shells, isolated or embedded in the multi-particle spectrum. To prove the convergence of Araki--Haag detectors on eigenstates of the mass operator $M=\sqrt{H^2-|\mathbf{P}|^2}$ is relatively simple (see Proposition~\ref{prop:ConvergenceMassEigenstates}). However, in general, $\mathcal{H}_{\mathrm{sc}}(P)$ may also contain exotic states for which we cannot prove the convergence of Araki--Haag detectors. In Lorentz-covariant quantum field theories, these exotic states correspond to the singular continuous spectrum of the mass operator.

Theorem~\ref{thm:MainResult} promises to advance our understanding of two-particle asymptotic completeness in local relativistic quantum field theory. A quantum field theory is asymptotically complete in the two-particle region if $E(\Delta)\mathcal{H} = E(\Delta)\mathcal{H}^\mathrm{out} = E(\Delta)\mathcal{H}^\mathrm{in}$ for all $\Delta \subset \sigma(P)$ between the two- and three-particle threshold. To prove two-particle asymptotic completeness, we may adopt the following strategy: According to the last statement of the theorem, Araki--Haag detectors map the orthogonal complement of scattering states to 0. In physically relevant quantum field theories, quantum states should be accessible through experiments, implying that we should be able to construct a detector capable of detecting a given state $\psi\in\mathcal{H}$. Notably, we can indeed construct such detectors for one-particle states, as demonstrated in Lemma~\ref{lem:InsertionAHDetector}. However, establishing this property for states in the multi-particle spectrum requires additional assumptions.

\begin{cor} \label{cor:AsymptoticCompleteness}
	Let $\Delta\subset\mathbb{R}^d$ be compact. If, for every $\psi\in E(\Delta)\mathcal{H}_\mathrm{ac}(P)$, a creation operator $B^*$ exists such that $B\psi$ is a one-particle state and
	\begin{align}
		\lim_{t\to\infty} \e^{\I t H} \int_{\mathbb{R}^s} (B^*B)(\mathbf{x}) \diff \mathbf{x} \e^{-\I t H} \psi \neq 0,
	\end{align}
	then $E(\Delta)\mathcal{H}_\mathrm{ac}(P) = E(\Delta)\mathcal{H}^\mathrm{out}$.
\end{cor}

It is worth noting that there exist non-trivial models to which Theorem~\ref{thm:MainResult} applies. As discussed above, the assumption $\psi\in\mathcal{H}_\mathrm{ac}(P)$ should not be very restrictive, but it may be difficult to verify it in models. Notably, it is known from \cite{spencer1975} that the energy-momentum spectrum of the weakly coupled $P(\phi)_2$ model is absolutely continuous in the two-particle region. However, it is also known that this model is asymptotically complete in the two-particle region, and the convergence of Araki--Haag detectors on scattering states is already known from \cite{araki1967}. Models of interest for applying Theorem~\ref{thm:MainResult} are those where the spectral assumption holds true, but asymptotic completeness either is not established or fails. Among the simplest models, which belong to this class, are certain generalised free fields. We discuss a model with non-trivial $S$-matrix in Section~\ref{ssec:Models}.

In the remainder of the introduction, we outline the proof strategy of Theorem~\ref{thm:MainResult}. The convergence of $C(h,t)$ on scattering states has been previously established in~\cite{araki1967}. Hence, we focus on proving the convergence of $C(h,t)$ on states orthogonal to all scattering states, similarly as in \cite{dybalski2018}. We formulate the convergence of $C(h,t)$ on the orthogonal complement of scattering states in Theorem~\ref{thm:NonDetectionNonScatteringStates}. Consequently, Theorem~\ref{thm:MainResult} directly follows from Theorem~\ref{thm:NonDetectionNonScatteringStates} (see Section~\ref{sec:ArakiHaagDetectors} for the proof of Theorem~\ref{thm:MainResult}).

The first step in proving Theorem~\ref{thm:NonDetectionNonScatteringStates} is to reduce the convergence of a single detector to that of two detectors by introducing a second auxiliary detector (see Lemma~\ref{lem:InsertionAHDetector}). The convergence of two detectors sensitive to particles with distinct velocities (i.e.~detectors for which the velocity functions $h_1$ and $h_2$ have disjoint support) has been analysed in \cite{dybalski2014}. However, in our case, the supports of the velocity functions intersect. To proceed, a novel result in the form of an improved convergence property of Haag--Ruelle scattering states is required, which we establish in Theorem~\ref{thm:L2Convergence}. Under the assumption that the momentum transfers of $B_1^*$ and $B_2^*$ are separated, we prove that, for every $\psi\in\mathcal{H}_{\mathrm{ac}}(P)$, the function
\begin{align} \label{eq:DetectorFunction}
	(\mathbf{x},\mathbf{y}) \mapsto \scp{\psi}{\varphi_t(\mathbf{x},\mathbf{y})}=\e^{-\I t(\omega(D_\mathbf{x})+\omega(D_\mathbf{y}))}\scp{\e^{-\I t H}\psi}{B_1^*(\mathbf{x})B_2^*(\mathbf{y})\Omega}
\end{align}
converges in $L^2(\mathbb{R}^{2s})$ as $t\to\infty$, where $\omega=\sqrt{m^2+|\cdot|^2}$ is the relativistic dispersion relation and $D_\mathbf{x}=-\I \partial_{\mathbf{x}}$. Additionally, we demonstrate that the limit is 0 if $\psi$ lies in the orthogonal complement of the scattering states.

The proof strategy of Theorem~\ref{thm:L2Convergence} resembles the proof of Lavine's Theorem for the existence and completeness of two-particle wave operators in quantum mechanics \cite[Theorem~7.1.4]{amrein1996} \cite[Proposition~7.2]{amrein2009}. In the first step, we apply Cook's method and reformulate the resulting expression in relative coordinates $(\mathbf{u}=\mathbf{x}-\mathbf{y},\mathbf{v}=(\mathbf{x}+\mathbf{y})/2)$. We prove that it is sufficient to establish the $L^2(\mathbb{R}^{2s})$-convergence of the following function to 0 as $t\to\infty$:
\begin{align}\label{eq:FunctionInRelativeCoordinates}
	(\mathbf{u},\mathbf{v}) \mapsto \int_{t}^{\infty} \e^{-\I \tau (\omega(\frac{1}{2}D_\mathbf{v}+D_{\mathbf{u}}) + \omega(\frac{1}{2}D_\mathbf{v}-D_{\mathbf{u}}))} \scp{\psi}{\e^{\I\tau H} \e^{-\I\mathbf{v}\cdot\mathbf{P}}\phi(\mathbf{u})} \diff \tau,
\end{align}
where $\phi$ is a Hilbert space-valued Schwartz function. In the next step, we perform a fibre decomposition of \eqref{eq:FunctionInRelativeCoordinates} along the total momentum $D_{\mathbf{v}}$ by taking the Fourier transformation. This step is similar to removing the centre-of-mass motion in many-body problems. We denote the Fourier variable corresponding to $\mathbf{v}$ by $\mathbf{p}$, and we arrive at the following bound for the $L^2$-norm of \eqref{eq:FunctionInRelativeCoordinates}:
\begin{align} \label{eq:CrucialBound}
	&\int_{K_{\mathrm{tot}}} \left( \sup_{\|f\|_{L^2}=1} \int_{t}^{\infty} \|\jap{A_{\mathbf{p}}}^{-\nu} \e^{\I \tau \omega_\mathbf{p}(D_{\mathbf{u}})} \chi f\|^2_{L^2} \diff\tau \right) \notag \\
	&\times \left(\int_{t}^{\infty} \|\jap{A_{\mathbf{p}}}^\nu \mathcal{F}_{\mathbf{v}\to\mathbf{p}}\scp{\psi}{\e^{\I\tau H} \e^{-\I\mathbf{v}\cdot\mathbf{P}}\phi} \|^2_{L^2} \diff \tau \right) \diff \mathbf{p},
\end{align}	
where $K_{\mathrm{tot}}$ is a compact set, $\jap{\cdot}=\sqrt{1+|\cdot|^2}$ denotes the Japanese bracket, $A_\mathbf{p}$ is a modified dilation operator (see \eqref{eq:ConjugateOperator}),
\begin{align}
	\omega_{\mathbf{p}}(D_{\mathbf{u}}) = \omega(\mathbf{p}/2-D_{\mathbf{u}}) + \omega(\mathbf{p}/2+D_{\mathbf{u}})
\end{align} 
is a pseudo-differential operator that corresponds to the energy of two free particles with relativistic dispersion relation and total momentum $\mathbf{p}$, $\chi$ is a cut-off that projects out contributions with vanishing relative momentum, and $\mathcal{F}$ denotes the Fourier transformation.

Applying techniques from Mourre's conjugate operator method, we prove that the first factor in brackets in \eqref{eq:CrucialBound} is uniformly bounded in $t$. The second factor in brackets converges to 0 as a consequence of the microcausality axiom. This crucial step combines methods from quantum mechanics (Mourre theory) with concepts from quantum field theory (microcausality).

The paper is structured as follows: In Section~\ref{sec:HRScatteringTheory}, we summarise the assumptions of the paper and recall relevant facts from Haag--Ruelle scattering theory. We introduce Araki--Haag detectors in Section~\ref{sec:ArakiHaagDetectors}, where we also prove the convergence of Araki--Haag detectors on the orthogonal complement of scattering states (Theorem~\ref{thm:NonDetectionNonScatteringStates}). In Section~\ref{sec:L2Convergence}, we analyse the convergence of the function~\eqref{eq:DetectorFunction}. In Section~\ref{sec:Outlook}, we discuss applications of Theorem~\ref{thm:MainResult} for a potential proof of asymptotic completeness and the applicability of our results to models. Moreover, we provide an outlook for further research directions. In Appendix~\ref{appx:ConjugateOperatorMethod}, we review key results of Mourre's conjugate operator method, which are relevant for the paper, and develop the notion of locally smooth operators for a family of commuting self-adjoint operators.

\section{Haag--Ruelle scattering theory}
\label{sec:HRScatteringTheory}

In this paper, we work within the framework of Haag--Kastler quantum field theory. We summarise our assumptions and notation in Section~\ref{ssec:Assumptions}. Following this, we revisit some well-known facts from Haag--Ruelle scattering theory. We refer to \cite[Section~2.2, Section~6]{dybalski2014_2} for more details. 

\subsection{Assumptions and notation}
\label{ssec:Assumptions}

Observables are described by a net of $C^*$-algebras $\{\mathcal{A}(O)\}_{O\in \mathcal{J}}$, where $\mathcal{J}$ is the set of all bounded open subsets of $\mathbb{R}^d = \mathbb{R}^{1+s}$ (with $d$ representing the spacetime dimension and $s=d-1$ the spatial dimension). If $O_1 \subset O_2$, then $\mathcal{A}(O_1) \subset \mathcal{A}(O_2)$ (isotony), and if $O_1$ is contained in the causal complement of $O_2$, then $[\mathcal{A}(O_1),\mathcal{A}(O_2)]=\{0\}$ (microcausality). The global algebra $\mathcal{A}$ is the inductive $C^*$-limit of the net $\{\mathcal{A}(O)\}_{O\in\mathcal{J}}$. Moreover, a morphism $\alpha: \mathbb{R}^d \to \Aut(\mathcal{A})$ from the translation group $\mathbb{R}^d$ to the automorphism group $\Aut(\mathcal{A})$ of $\mathcal{A}$ exists such that, for $O\in\mathcal{J}$ and $x\in\mathbb{R}^d$, $\alpha_x \mathcal{A}(O) = \mathcal{A}(O+x)$.

We assume that a vacuum state $\rho_0 : \mathcal{A} \to \mathbb{C}$ invariant under the action of $\alpha$ exists. The GNS triple of $\rho_0$ is denoted by $(\pi_0, \mathcal{H}, \Omega)$, where $\pi_0: \mathcal{A} \to \mathfrak{B}(\mathcal{H})$ is a representation of $\mathcal{A}$ on the Hilbert space $\mathcal{H}$ and $\Omega$ is the (cyclic) vacuum vector. Moreover, we impose the following assumptions: The spacetime translations are unitarily implemented on $\mathcal{H}$, that is, a strongly continuous unitary representation $U: \mathbb{R}^d \to \mathcal{B}(\mathcal{H})$ exists such that, for $A \in \mathcal{A}$ and $x\in\mathbb{R}^d$, $\pi_0(\alpha_x A) = U(x) \pi_0(A) U(x)^*$. Up to a multiplicative constant, $\Omega$ is the unique vector such that, for all $x\in\mathbb{R}^d$, $U(x) \Omega = \Omega$. The representation $U$ satisfies the strong spectrum condition, that is, if $\sigma(P)$ denotes the joint spectrum of the energy-momentum operator $P=(H,\mathbf{P})$ (i.e.~the generators of $U$), then $\{0\} \cup H_m \subset \sigma(P) \subset \{0\} \cup H_m \cup G_{2m}$, where 
\begin{align}
	H_m &= \{ (p_0,\mathbf{p}) \in \mathbb{R}^d \mid p_0 = \omega(\mathbf{p}) = \sqrt{m^2 + |\mathbf{p}|^2} \}, \\
	G_{2m} &= \{ (p_0,\mathbf{p}) \in \mathbb{R}^d \mid p_0 \geq \sqrt{(2m)^2 + |\mathbf{p}|^2} \}
\end{align}
are the mass hyperboloid of mass $m>0$ and the multi-particle spectrum, respectively. For simplicity, we assume that the mass eigenspace $\mathfrak{h}_m = E(H_m)\mathcal{H}$ of the mass operator $M=\sqrt{H^2-|\mathbf{P}|^2}$ corresponds to a single spinless particle, where $E$ is the spectral measure of $P$. 

The local observable algebras $\mathcal{R}(O)$, $O\in\mathcal{J}$, are von Neumann algebras generated by $\pi_0(\mathcal{A}(O))$, and the von Neumann algebra $\mathcal{R}$ generated by $\pi_0(\mathcal{A})$ is the global observable algebra.\footnote{Under our assumptions, it can be shown that $\mathcal{R}=\mathfrak{B}(\mathcal{H})$ \cite[Theorem~4.6]{araki1999}.} For $A\in\mathcal{R}$ and $x=(x_0,\mathbf{x}) \in \mathbb{R}^{1+s}$, we write $A(x)=U(x)AU(x)^*$ and abbreviate $A(0,\mathbf{x})$ by $A(\mathbf{x})$. If $f\in \mathcal{S}(\mathbb{R}^d)$ is a Schwartz function, we define $A(f) =\int_{\mathbb{R}^d} f(x) A(x) \diff x$, where the integral is defined in the weak sense. Similarly, for a Schwartz function $f\in\mathcal{S}(\mathbb{R}^s)$, we write $A[f]=\int_{\mathbb{R}^s} f(\mathbf{x}) A(\mathbf{x}) \diff \mathbf{x}$.

\subsection{Haag--Ruelle creation operators}

We explain how to define creation operators $B^*\in\mathcal{R}$ such that $B^*\Omega \in \mathfrak{h}_m$ is a one-particle state with good localisation properties. It is untenable to take $B^*$ from a local observable algebra. Instead, we choose $B^*$ to be almost local, that is, $B^*$ is essentially localised in a double cone, and its norm outside a double cone decays rapidly.

\begin{defn}\label{defn:AlmostLocal}
	Let $K_r$ be the double cone of radius $r>0$. An element $A\in\mathcal{R}$ is \textbf{almost local} if a sequence $(A_r)$ of local operators $A_r\in\mathcal{R}(K_r)$ exists such that $A_r$ converges rapidly in norm to $A$ as $r\to\infty$ (i.e., for every $N\in\mathbb{N}$, $\|A-A_r\| \leq C_N r^{-N}$).
\end{defn}

According to the microcausality axiom, the commutator of two observables localised in space-like separated regions vanishes. The commutator of two almost local observables does not necessarily vanish, but its norm decays rapidly with increasing space-like separation.

\begin{lem}\label{lem:AlmostLocalCommutatorDecay}
	If $A_1, A_2 \in\mathcal{R}$ are almost local, then, for every $N\in\mathbb{N}$, a constant $C_N$ exists such that $\|[A_1,A_2(\mathbf{x})]\| \leq C_N \jap{\mathbf{x}}^{-N}$.
\end{lem}

The lemma is a direct consequence of Definition~\ref{defn:AlmostLocal}. Next, we introduce the energy-momentum transfer of an element $A\in\mathcal{R}$, which characterises the change in energy-momentum of a state upon the action of $A$.

\begin{defn}
	The \textbf{energy-momentum transfer} (or \textbf{Arveson spectrum}) $\ArSp{A}$ of an element $A\in\mathcal{R}$ is the support of the operator-valued distribution
	\begin{align}
		\mathbb{R}^d \ni p \mapsto \check{A}(p) = \frac{1}{(2\pi)^d} \int_{\mathbb{R}^d} \e^{-\I p\cdot x} A(x) \diff x,
	\end{align}
	where $p\cdot x = p_0 x_0 - \mathbf{p}\cdot \mathbf{x}$ denotes the Minkowski product. Moreover, the \textbf{momentum transfer} of $A$ is $\pi_{\mathbf{P}}(\ArSp{A})=\{ \mathbf{p} \in \mathbb{R}^s \mid \exists p_0 \in \mathbb{R} \colon (p_0, \mathbf{p}) \in \ArSp{A} \}$.
\end{defn}

We list the following key properties of the energy-momentum transfer: If $A^*$ denotes the adjoint of $A\in\mathcal{R}$, then $\sigma_\alpha(A^*)=-\sigma_\alpha(A)$. If $x\in\mathbb{R}^d$, then $\ArSp{A(x)}=\ArSp{A}$. For a Schwartz function $f\in \mathcal{S}(\mathbb{R}^d)$, it holds that $\sigma_\alpha(A(f))\subset\supp(\hat{f})\cap \ArSp{A}$, where $\hat{f}$ is the Fourier transform of $f$. Furthermore, if $\sigma_\alpha(A)$ is compact and $\hat{f}\in C_c^\infty(\mathbb{R}^d)$ satisfies $\hat{f}=1$ on $\sigma_\alpha(A)$, then $A=A(f)$. Similarly, if $\pi_{\mathbf{P}}(\ArSp{A})$ is compact and $\hat{f}\in C_c^\infty(\mathbb{R}^s)$ satisfies $\hat{f}=1$ on $\pi_{\mathbf{P}}(\ArSp{A})$, then $A=A[f]$.

The following proposition justifies the name energy-momentum transfer for the set $\sigma_\alpha(A)$.

\begin{prop}[{\cite[(2.4)]{dybalski2014}}]\label{prop:EMTransfer}
	If $\Delta \subset\mathbb{R}^d$ is a Borel set, then 
	\begin{align}
		A E(\Delta)\mathcal{H} \subset E(\overline{\Delta+\sigma_\alpha(A)})\mathcal{H}.
	\end{align}
\end{prop}

We see that $B^*\Omega$ is a one-particle state with good localisation properties if $B^*$ is almost local and has energy-momentum transfer contained in a sufficiently small set that intersects the mass hyperboloid $H_m$. This motivates the following definition:

\begin{defn}
	An element $B^* \in \mathcal{R}$ is a \textbf{creation operator} if $B^*$ is almost local, its energy-momentum transfer $\ArSp{B^*}$ is a compact subset of the closed forward light cone $V_+=\{p\in\mathbb{R}^d\mid p^0 \geq |\mathbf{p}|\}$, and $\emptyset \neq \ArSp{B^*} \cap \sigma(P) \subset H_m$.
\end{defn}

In scattering theory, typically the time evolution in the distant past and far future of the (interacting) Hamiltonian $H$ is compared with the time evolution of a simpler (free) system. In Haag--Ruelle scattering theory, we compare the time evolution generated by $H$ with the evolution of Klein--Gordon wave packets by forming time-dependent \textbf{Haag--Ruelle creation operators}:
\begin{align}
	B^*_t [f_t] = \int_{\mathbb{R}^s} f_t(\mathbf{x}) B^*_t(\mathbf{x}) \diff \mathbf{x}, \ \ B_t^*(\mathbf{x}) = U(t,\mathbf{x}) B^* U(t,\mathbf{x})^*, \ f_t = \e^{-\I t \omega(D_{\mathbf{x}})} f,
\end{align}
where $B^*$ is a creation operator, $f \in \mathcal{S}(\mathbb{R}^s)$ is a Schwartz function with compact support in Fourier space, and $\omega(D_\mathbf{x}) = \sqrt{m^2 + |D_\mathbf{x}|^2}$ is the relativistic dispersion relation (i.e.~$f_t$ is a regular positive energy solution of the Klein--Gordon equation). The definition of Haag--Ruelle creation operators can be extended to $f\in L^2(\mathbb{R}^s)$. To define this extension, the following estimate is required, which relies on a uniform bound by Buchholz \cite{buchholz1990}.

\begin{prop}[{\cite[Lemma~3.4]{dybalski2014}}]\label{prop:HRCreationOperatorL2Bound}
	Let $B\in\mathcal{R}$ be almost local such that $\sigma_\alpha(B)\cap V_+=\emptyset$. For every compact subset $\Delta\subset\mathbb{R}^d$, a constant $C_\Delta$ exists such that, for every $\psi\in\mathcal{H}$,
	\begin{align}
		\int_{\mathbb{R}^s} \|B(\mathbf{x})E(\Delta)\psi\|^2 \diff\mathbf{x} \leq C_\Delta \|\psi\|^2.
	\end{align}
\end{prop}

The following two corollaries are consequences of Proposition~\ref{prop:EMTransfer} and Proposition~\ref{prop:HRCreationOperatorL2Bound}. The first corollary defines $B^*[f]$ on states of bounded energy for $f\in L^2(\mathbb{R}^s)$.

\begin{cor}[{\cite[Lemma~6.4]{dybalski2014_2}}]\label{cor:HRCreationOperatorL2Bound}
	Let $B^*$ be a creation operator. For every compact subset $\Delta\subset\mathbb{R}^d$, a constant $C_\Delta$ exists such that $\|B[f] E(\Delta)\| \leq C_{\Delta} \|f\|_{L^2}$ and $\|B^*[f] E(\Delta)\| \leq C_{\Delta} \|f\|_{L^2}$.
\end{cor}

\begin{cor}[{\cite[Lemma~4]{buchholz1974}}] \label{cor:HRCreationOperatorL2Bound2}
	Let $B_1^*,\dots,B_n^*$ be creation operators. A constant $C<\infty$ exists such that, for all $f\in L^2(\mathbb{R}^{ns})$,
	\begin{align}
		\norm{\int_{\mathbb{R}^{ns}} f(\mathbf{x}_1,\dots,\mathbf{x}_n) B_1^*(\mathbf{x}_1)\dots B_n^*(\mathbf{x}_n)\Omega \diff\mathbf{x}_1\dots\diff\mathbf{x}_n} \leq C \|f\|_{L^2}. 
	\end{align}
\end{cor}

\subsection{Scattering states}
\label{ssec:ScatteringStates}

If a Haag--Ruelle creation operator $B_t^*[f_t]$ is applied to the vacuum vector $\Omega$, the interacting and free time evolution exactly cancel each other, and we obtain a time-independent one-particle state.
\begin{lem}	\label{lem:HRCreationOperatorOneParticleState}
	Let $B^*$ be a creation operator and $f\in L^2(\mathbb{R}^s)$. The one-particle state $B_t^*[f_t] \Omega = \hat{f}(\mathbf{P}) B^* \Omega \in \mathfrak{h}_m$ is independent of $t$. 
\end{lem}

The next theorem provides the construction of multi-particle scattering states. We state the theorem only for outgoing scattering states (i.e.~for $t\to\infty$). The results for incoming scattering states (i.e.~$t\to -\infty$) are similar. 

\begin{thm}
	\label{thm:ScatteringStates}
	Let $B_1^*,\dots,B_n^*$ be creation operators and $f_1, \dots, f_n\in L^2(\mathbb{R}^s)$. The scattering states
	\begin{align}
		\psi_1 \outgoing \cdots \outgoing \psi_n = \lim_{t \to \infty} B^*_{1,t}[f_{1,t}] \dots B^*_{n,t}[f_{n,t}] \Omega
	\end{align}
	exist and depend only on the one-particle states $\psi_i = B^*_i[f_i]\Omega$. Moreover, scattering states have the following properties: 
	\begin{enumerate}
		\item For $x \in \mathbb{R}^d$,
		\begin{align}\label{eq:ScatteringPoincareCovariance}
			U(x) (\psi_1 \outgoing \cdots \outgoing \psi_n) = U(x) \psi_1 \outgoing \cdots \outgoing U(x) \psi_n.
		\end{align}
		\item If $\phi_1 \outgoing \cdots \outgoing \phi_m$ is another scattering state, then		
		\begin{align}\label{eq:ScatteringFockSpaceStructure}
			\scp{\psi_1 \outgoing \cdots \outgoing \psi_n}{\phi_1 \outgoing \cdots \outgoing \phi_m}
			= \delta_{nm} \sum_{\sigma \in S_n} \scp{\psi_1}{\phi_{\sigma(1)}} \dots \scp{\psi_n}{\phi_{\sigma(n)}},
		\end{align}
		where $S_n$ is the group of permutations of $n$ elements.
	\end{enumerate}	
\end{thm}

\begin{rem}
	In \cite[Theorem~6.5]{dybalski2014_2}, the above theorem is stated for Klein--Gordon wave packets $f_1, \dots, f_n\in L^2(\mathbb{R}^s)$ with disjoint velocity support. The span of $f_1\otimes \cdots \otimes f_n$, where $f_1,\dots,f_n$ have disjoint velocity support, is dense in $L^2(\mathbb{R}^{ns})$. By Corollary~\ref{cor:HRCreationOperatorL2Bound2}, the theorem extends to arbitrary families of Klein--Gordon wave packets.
\end{rem}

Let $\Gamma(\mathfrak{h}_m)$ be the symmetric Fock space over the one-particle space $\mathfrak{h}_m = E(H_m)\mathcal{H}$ with Fock vacuum $\Omega_0$, and, for $\psi \in \mathfrak{h}_m$, let $a^*(\psi)$ be the Fock creation operator such that $a^*(\psi)\Omega_0 = \psi$. Moreover, let $\mathcal{H}^{\mathrm{out}}$ be the Hilbert space generated by scattering states: 
\begin{align}
	\mathcal{H}^{\mathrm{out}} = \mathrm{span}\{ \Omega, \psi_1 \outgoing \dots \outgoing \psi_n \mid \psi_1, \dots, \psi_n \in \mathfrak{h}_m, n\in \mathbb{N} \}.
\end{align}
We define the isometric \textbf{wave operator} $W^\mathrm{out}:\Gamma(\mathfrak{h}_m) \to \mathcal{H}^{\mathrm{out}}$ by the following relations:
\begin{align}
	W^\mathrm{out}\Omega_0 &= \Omega, \\
	W^\mathrm{out}(a^*(\psi_1) \cdots a^*(\psi_n) \Omega_0) &= \psi_1 \outgoing \cdots \outgoing \psi_n.
\end{align}
In the following, we write $a_\mathrm{out}^*(\psi) = W^\mathrm{out} a^*(\psi) (W^\mathrm{out})^*$.

\subsection{Carleman functions} \label{ssec:CarlemanFunctions}

A weakly measurable Hilbert space-valued function $\varphi:\mathbb{R}^{ns}\to\mathcal{H}$ is a \textbf{Carleman function} if its Carleman norm
\begin{align}
	\|\varphi\|_{\mathrm{C}} &= \sup_{\norm{\psi}_{\mathcal{H}} = 1} \left(\int_{\mathbb{R}^{ns}} |\scp{\psi}{\varphi(\mathbf{x}_1,\dots,\mathbf{x}_n)}|^2 \diff \mathbf{x}_1 \dots \diff \mathbf{x}_n\right)^\frac{1}{2} \notag \\
	&= \sup_{\|f\|_{L^2}=1} \norm{\int_{\mathbb{R}^{ns}} f(\mathbf{x}_1,\dots,\mathbf{x}_n) \varphi(\mathbf{x}_1,\dots,\mathbf{x}_n) \diff \mathbf{x}_1 \dots \diff \mathbf{x}_n}
\end{align}
is finite.\footnote{In \cite[p.~63]{halmos1978}, such functions are called \textit{bounded} Carleman functions in the sense that $\varphi$ defines a kernel of a bounded operator mapping $L^2(\mathbb{R}^{ns})$ into $\mathcal{H}$.} Let $B_1^*,\dots,B_n^*$ be creation operators. By Corollary~\ref{cor:HRCreationOperatorL2Bound2}, $(\mathbf{x}_1,\dots,\mathbf{x}_n) \mapsto \varphi_0(\mathbf{x}_1,\dots,\mathbf{x}_n)=B_1^*(\mathbf{x}_1)\dots B_n^*(\mathbf{x}_n)\Omega$ is a Carleman function. In particular, for every $\psi \in \mathcal{H}$, $\scp{\psi}{\varphi_0}\in L^2(\mathbb{R}^{ns})$. We define $\varphi_t$ to be the Carleman function that obeys, for every $\psi\in\mathcal{H}$, the following identity:
\begin{align} \label{eq:varphit}
	\scp{\psi}{\varphi_t}=\e^{-\I t (\omega(D_{\mathbf{x}_1})+\cdots+\omega(D_{\mathbf{x}_n}))}\scp{\e^{-\I tH}\psi}{\varphi_0}.
\end{align}
Clearly, $\norm{\varphi_t}_{\mathrm{C}} = \norm{\varphi_0}_{\mathrm{C}}$. Moreover, if $f_1,\dots,f_n \in \mathcal{S}(\mathbb{R}^s)$ are Schwartz functions which satisfy $\hat{f}_i = 1$ on the momentum transfer $\pi_{\mathbf{P}}(\sigma_\alpha(B_i^*))$, then
\begin{align}
	\varphi_t(\mathbf{x}_1,\dots,\mathbf{x}_n) = B_{1,t}^*[f_{1,t}^{\mathbf{x}_1}] \dots B_{n,t}^*[f_{n,t}^{\mathbf{x}_n}]\Omega,
\end{align}
where $f_i^{\mathbf{x}_i} = f_i(\cdot - \mathbf{x}_i)$. By Theorem \ref{thm:ScatteringStates}, for every $\mathbf{x}_1,\dots,\mathbf{x}_n\in \mathbb{R}^s$, $\varphi_t(\mathbf{x}_1,\dots,\mathbf{x}_n)$ converges in $\mathcal{H}$ as $t\to\infty$ to the scattering state
\begin{align}
	\varphi_+(\mathbf{x}_1,\dots,\mathbf{x}_n) = a^*_\mathrm{out}(B_1^*(\mathbf{x}_1)\Omega) \dots a^*_{\mathrm{out}}(B_n^*(\mathbf{x}_n)\Omega)\Omega,
\end{align}
and $\norm{\varphi_+}_\mathrm{C} \leq \norm{\varphi_0}_\mathrm{C}$ by Fatou's lemma. In particular, for every $\psi\in\mathcal{H}$, $\scp{\psi}{\varphi_+}\in L^2(\mathbb{R}^{ns})$.

\section{Araki--Haag detectors} \label{sec:ArakiHaagDetectors}

We introduce Araki--Haag detectors in Section~\ref{ssec:ArakiHaagFormula}, where we review convergence results of Araki--Haag detectors on scattering states. Subsequently, in Section~\ref{ssec:ConvergenceAHDetector}, we discuss the convergence of single Araki--Haag detectors. Specifically, we prove our main result (Theorem~\ref{thm:MainResult}).

\subsection{Araki--Haag formula} \label{ssec:ArakiHaagFormula}

A detector $C$ is an almost local observable measuring deviations from the vacuum with $C(x) = U(x)CU(x)^*$ representing a measurement in the neighbourhood of the spacetime point $x \in\mathbb{R}^d$.

\begin{defn}
	A self-adjoint element $C\in\mathcal{R}$ is a \textbf{detector} if $C$ is almost local and $C\Omega = 0$.
\end{defn}

\begin{exmp}
	Let $B\in\mathcal{R}$ be almost local and denote the closed forward light cone by $V_+$. If $\sigma_\alpha(B)\cap V_+ = \emptyset$, then $C=B^*B$ is a detector. Indeed, $C$ is almost local and $B\Omega=0$ by Proposition~\ref{prop:EMTransfer} and the spectrum condition. The detectors of the form $B^*B$ generate a $*$-algebra $\mathcal{C}$, where each element is itself a detector.
\end{exmp}

For an almost local element $A\in \mathcal{R}$ and regular scattering states $\phi,\psi \in \mathcal{H}^{\mathrm{out}}$ of bounded energy in which no pair of particles has the same velocity, Araki and Haag \cite[Theorem~2]{araki1967} proved the following asymptotic expansion as $t\to\infty$:
\begin{align}\label{eq:AsymptoticExpansion}
	\scp{\phi}{A(t,\mathbf{x})\psi} &= \scp{\Omega}{A\Omega} \scp{\phi}{\psi} + \int_{\mathbb{R}^s} \left( \diracscp{\mathbf{p}}{A(t,\mathbf{x})|\Omega} \scp{\phi}{a^*_{\mathrm{out}}(\mathbf{p})\psi} +
	\diracscp{\Omega|A(t,\mathbf{x})}{\mathbf{p}} \scp{\phi}{a_{\mathrm{out}}(\mathbf{p})\psi} \right) \diff\mathbf{p} \notag \\
	&+ \int_{\mathbb{R}^s} \int_{\mathbb{R}^s} \diracscp{\mathbf{q}}{A(t,\mathbf{x})|\mathbf{p}} \scp{\phi}{a_{\mathrm{out}}^*(\mathbf{q}) a_{\mathrm{out}}(\mathbf{p}) \psi} \diff\mathbf{p} \diff\mathbf{q} + R_{\phi,\psi,A}(t,\mathbf{x}),
\end{align}
where $R_{\phi,\psi,A}(t,\mathbf{x})$ is a remainder that decays rapidly in $t$ uniformly in $\mathbf{x}\in\mathbb{R}^s$. Here, we identify elements of the one-particle space $\mathfrak{h}_m$ with wave functions in $L^2(\mathbb{R}^s)$. A single-particle state with momentum $\mathbf{p}$ is denoted as $|\mathbf{p}\rangle$ with normalisation $\diracscp{\mathbf{p}}{\mathbf{q}} = \delta(\mathbf{p}-\mathbf{q})$. Remember that, according to our assumptions, the one-particle space $\mathfrak{h}_m$ describes a single spinless particle (see \cite{araki1967} for the asymptotic expansion in the general case). 

If $A=C$ is a detector, the first two terms of the asymptotic expansion \eqref{eq:AsymptoticExpansion} vanish, and the dominant contributions of $\langle \phi | C(t,\mathbf{x}) | \psi \rangle$ as $t\to\infty$ arise from single-particle excitations. It can be shown that $\scp{\phi}{C(t,\mathbf{x})\psi}$ converges to 0 with the rate $t^{-s}$ due to the dispersion of quantum states. To obtain a non-trivial limit as $t\to\infty$, we integrate the detector $C(t,\mathbf{x})$ over the entire space: 
\begin{align}
	C(h,t) = \int_{\mathbb{R}^s} h\left(\frac{\mathbf{x}}{t}\right) C(t,\mathbf{x}) \diff \mathbf{x}, \ h\in C_c^\infty(\mathbb{R}^s).
\end{align}
Araki and Haag \cite[Theorem~4]{araki1967} proved that, for scattering states $\phi,\psi \in \mathcal{H}^{\mathrm{out}}$ as above, 
\begin{align}\label{eq:ArakiHaagFormula}
	\lim_{t\to\infty} \int_{\mathbb{R}^s} h\left(\frac{\mathbf{x}}{t}\right) \scp{\phi}{C(t,\mathbf{x}) \psi} \diff \mathbf{x} 
	= (2\pi)^s \int_{\mathbb{R}^s} h(\nabla\omega(\mathbf{p})) \diracscp{\mathbf{p}}{C|\mathbf{p}} \scp{\phi}{a_{\mathrm{out}}^*(\mathbf{p}) a_{\mathrm{out}}(\mathbf{p}) \psi} \diff \mathbf{p}.
\end{align}

\begin{rem}
	If $C\in\mathcal{C}$, we can extend the convergence result to all scattering states of bounded energy by Proposition~\ref{prop:HRCreationOperatorL2Bound}. 
\end{rem}

The asymptotic observable \eqref{eq:ArakiHaagFormula} resembles the Fock space number operator (i.e.~a particle counter). The additional factor $h(\nabla\omega(\mathbf{p})) \diracscp{\mathbf{p}}{C|\mathbf{p}}$ is interpreted as the sensitivity of the counter to measure a particle of momentum $\mathbf{p}$. Specifically, $h$ is a velocity filter because particles with velocity $\nabla \omega(\mathbf{p})$ outside the support of $h$ are not counted. Henceforth, we refer to these asymptotic observables as \textbf{Araki--Haag detectors}. The formula \eqref{eq:ArakiHaagFormula} generalises to multiple detectors.
\begin{thm}[Araki--Haag formula, {\cite[Theorem~5]{araki1967}}] \label{thm:ArakiHaagFormula}
	Let $\phi,\psi \in \mathcal{H}^{\mathrm{out}}$ be scattering states of bounded energy and $C_1, \dots, C_n\in\mathcal{C}$. If $h_1, \dots, h_n \in C_c^\infty(\mathbb{R}^s)$ have disjoint support, then
	\begin{align}
		\lim_{t\to\infty} \scp{\phi}{C_1(h_1,t) \dots C_n(h_n,t) \psi} = 
		\int_{\mathbb{R}^{ns}} h_1(\nabla\omega(\mathbf{p}_1)) \dots h_n(\nabla\omega(\mathbf{p}_n)) \Gamma(\mathbf{p}_1, \dots, \mathbf{p}_n) \notag \\ 
		\times \scp{\phi}{a_{\mathrm{out}}^*(\mathbf{p}_1)a_{\mathrm{out}}(\mathbf{p}_1)\dots a_{\mathrm{out}}^*(\mathbf{p}_n) a_{\mathrm{out}}(\mathbf{p}_n)\psi} \diff\mathbf{p}_1 \dots \diff\mathbf{p}_n,
	\end{align}
	where $\Gamma(\mathbf{p}_1, \dots, \mathbf{p}_n) = (2\pi)^{ns} \diracscp{\mathbf{p}_1}{C_1|\mathbf{p}_1} \dots \diracscp{\mathbf{p}_n}{C_n|\mathbf{p}_n}$.
\end{thm}

\subsection{Convergence of single Araki--Haag detectors} \label{ssec:ConvergenceAHDetector}

To obtain particle detectors that are sensitive to particles of mass $m$, we choose $C=B^*B$, where $B^*$ is a creation operator. However, such detectors may also be sensitive to bound states of other masses as the following proposition illustrates.
\begin{prop}\label{prop:ConvergenceMassEigenstates}
	Let $\Delta \subset\mathbb{R}^d$ be compact. If $\psi \in E(\Delta)\mathcal{H}$ is an eigenvector of the mass operator $M = \sqrt{H^2-|\mathbf{P}|^2}$, then
	\begin{align}
		\int_{\mathbb{R}^s} \scp{\e^{-\I t H}\psi}{(B^*B)(\mathbf{x}) \e^{-\I t H}\psi} \diff \mathbf{x} = \int_{\mathbb{R}^s} \scp{\psi}{(B^*B)(\mathbf{x}) \psi} \diff \mathbf{x}.
	\end{align}
\end{prop}

\begin{proof}
	If $m_b$ is the eigenvalue corresponding to the eigenvector $\psi$, then $H\psi = (m_b^2+|\mathbf{P}|^2)^{1/2}\psi = \omega_{m_b}(\mathbf{P})\psi$, and
	\begin{align}\label{eq:DetectorEmbeddedMassShell}
		\int_{\mathbb{R}^s} \scp{\e^{-\I t H}\psi}{(B^*B)(\mathbf{x}) \e^{-\I t H} \psi} \diff \mathbf{x} 
		&= \scp{\psi}{\e^{\I t \omega_{m_b}(\mathbf{P})} \int_{\mathbb{R}^s} (B^*B)(\mathbf{x}) \diff \mathbf{x} \e^{-\I t \omega_{m_b}(\mathbf{P})} \psi} \notag \\
		&= \int_{\mathbb{R}^s} \scp{\psi}{(B^*B)(\mathbf{x}) \psi} \diff \mathbf{x}.
	\end{align}
	Observe that the translation-invariant operator $E(\Delta) \int_{\mathbb{R}^s} (B^*B)(\mathbf{x})\diff \mathbf{x} E(\Delta)$, which is well-defined by Proposition~\ref{prop:HRCreationOperatorL2Bound}, commutes with $\e^{\I t \omega_{m_b}(\mathbf{P})}$.
\end{proof}

We may exclude bound states of the mass operator by requiring that $\psi$ belongs to the jointly absolutely continuous spectral subspace $\mathcal{H}_{\mathrm{ac}}(P)$ of the energy-momentum operator $P$.\footnote{In~\cite{dybalski2014}, a modified version of Araki--Haag detectors was proposed to circumvent the detection of bound states.} In fact, mass shells, isolated or embedded in the multi-particle spectrum, belong to the singular continuous subspace $\mathcal{H}_\mathrm{sc}(P)$. If $M\psi=m_b\psi$, then $\psi= E(H_{m_b})\psi$, and the mass hyperboloid $H_{m_b}\subset \mathbb{R}^d$ has Lebesgue measure 0. However, in general, $\mathcal{H}_\mathrm{sc}(P)$ may include exotic states that are not bound states of the mass operator. Under the additional assumption of Lorentz covariance, it can be shown that $\mathcal{H}_\mathrm{ac}(P) = \mathcal{H}_\mathrm{ac}(M)$ or, equivalently, $\mathcal{H}_\mathrm{pp}(P)\oplus \mathcal{H}_\mathrm{sc}(P) = \mathcal{H}_\mathrm{pp}(M)\oplus \mathcal{H}_\mathrm{sc}(M)$. This identity implies that the exotic states in $\mathcal{H}_\mathrm{sc}(P)$ correspond to the singular continuous spectrum of the mass operator.

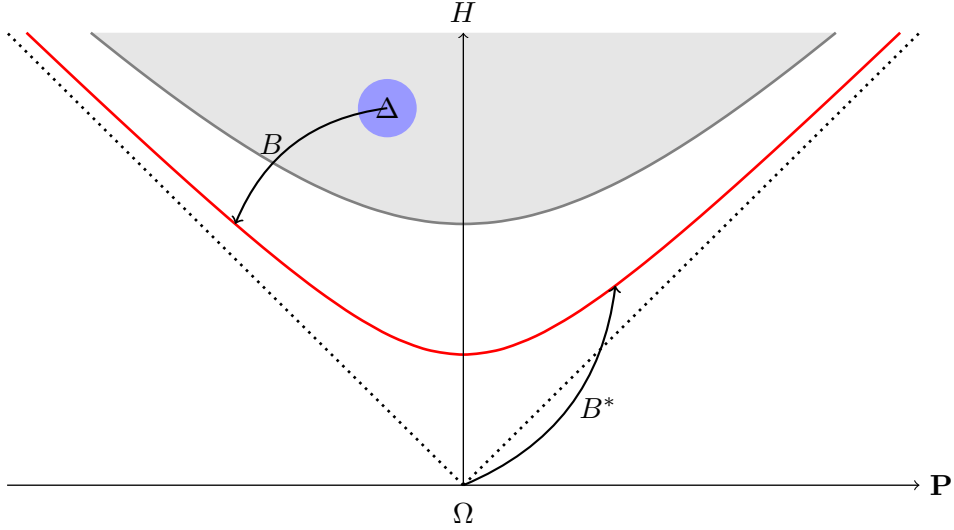
\begin{figure}
	\begin{center}
		\begin{tikzpicture}	
			\fill [gray!20, domain=-4.899:4.899, variable=\x] (4.899,6) -- plot ({\x}, {sqrt(\x*\x+12)}) -- (-4.899,6) -- cycle;
			
			\draw [color=black, dotted, line width=1pt] (0,0) -- (6,6);
			\draw [color=black, dotted, line width=1pt] (0,0) -- (-6,6);

			\node[circle,fill=blue!40] (c) at (-1,5) {$\Delta$};
			
			\draw[domain=-5.745:5.745, smooth, variable=\x, red, line width=1pt] plot ({\x}, {sqrt(\x*\x+3)});
			\draw[domain=-4.899:4.899, smooth, variable=\x, gray, line width=1pt] plot ({\x}, {sqrt(\x*\x+12)});
			
			\draw [color=black, ->, line width=0.5pt] (0,0) -- (0,6) node [above] {$H$};
			\draw [color=black, ->, line width=0.5pt] (-6,0) -- (6,0) node [right] {$\mathbf{P}$};
			\node[below=3pt] at (0,0){$\Omega$};
			
			\path (0,0) edge[bend right, ->, line width=0.8pt] node[right] {$B^*$} (2,2.646);
			\path (-1,5) edge[bend right, ->, line width=0.8pt] node[left] {$B$} (-3,3.464);
		\end{tikzpicture}
	\end{center}
	\caption{The energy-momentum spectrum contains an isolated mass shell (red line) and a continuous multi-particle spectrum (grey area) above the two-particle threshold (grey line). A creation operator $B^*$ maps the vacuum vector $\Omega$ to a one-particle state. In Theorem~\ref{thm:NonDetectionNonScatteringStates}, we assume that $B^*$ is such that $\overline{\Delta-\ArSp{B^*}}\cap\sigma(P)$ is a subset of the mass shell.}
	\label{fig:Spectrum}
\end{figure}

We aim to extend the Araki--Haag formula to arbitrary states of bounded energy (i.e.~to states that are not necessarily scattering states). We expect that Araki--Haag detectors do not detect a state that is orthogonal to all scattering states (i.e.~the limit \eqref{eq:ArakiHaagFormula} should be 0 if $\psi$ is orthogonal to all scattering states). We manage to prove this expectation for detectors $C=B^*B$, where $B^*$ is a creation operator, and for states $\psi\in \mathcal{H}_{\mathrm{ac}}(P)$ such that $B\psi$ is a one-particle state. 

Intuitively, the latter condition selects states of the multi-particle spectrum below the three-particle threshold. Specifically, if the energy-momentum spectrum of $\psi$ and the energy-momentum transfer of $B^*$ were point-like, then the condition that $B\psi$ is a one-particle state is always satisfied for creation operators $B^*$ and states $\psi$ below the three-particle threshold. However, due to the finite extension of the energy-momentum spectrum of $\psi$ and the energy-momentum transfer of $B^*$, it may happen that $B\psi$ has a component in the multi-particle spectrum for a state $\psi$ below the three-particle threshold. Also, there are creation operators $B^*$ and states $\psi$ above the three-particle threshold such that $B\psi$ is a one-particle state.

\begin{thm}\label{thm:NonDetectionNonScatteringStates}
	Let $\Delta \subset \mathbb{R}^d$ be compact, $\psi \in E(\Delta)\mathcal{H} \cap \mathcal{H}_{\mathrm{ac}}(P) \cap (\mathcal{H}^{\mathrm{out}})^\perp$, and $B^*$ a creation operator. If $\overline{\Delta - \sigma_\alpha(B^*)} \cap \sigma(P) \subset H_m$, then, for every $h\in L^\infty(\mathbb{R}^s)$,
	\begin{align}
		\lim_{t\to\infty} \e^{\I t H} \int_{\mathbb{R}^s} h\left(\frac{\mathbf{x}}{t}\right) (B^*B)(\mathbf{x}) \diff \mathbf{x} \e^{-\I t H}\psi = 0.
	\end{align}
\end{thm}

Before we present the proof, we explain that Theorem~\ref{thm:NonDetectionNonScatteringStates} implies our main result Theorem~\ref{thm:MainResult}. This follows from the fact the convergence of Araki--Haag detectors on scattering states is already known and that it suffices to prove convergence separately on $\mathcal{H}^\mathrm{out}$ and $(\mathcal{H}^\mathrm{out})^\perp$, as in \cite{dybalski2018}.

\begin{proof}[Proof of Theorem~\ref{thm:MainResult}]
	We decompose $\psi \in E(\Delta)\mathcal{H} \cap \mathcal{H}_{\mathrm{ac}}(P)$ into $\psi=\psi^\mathrm{out}+\psi^\perp$, where $\psi^\mathrm{out} \in \mathcal{H}^\mathrm{out}$ and $\psi^\perp \in (\mathcal{H}^\mathrm{out})^\perp$. The strong convergence of $C(h,t)\psi^\mathrm{out}$ in $\mathcal{H}$ was proved in \cite[Proposition~7.1]{dybalski2014_2}. The convergence of $C(h,t)\psi^\perp$ follows from Theorem~\ref{thm:NonDetectionNonScatteringStates}.
\end{proof}

\begin{proof}[Proof of Theorem~\ref{thm:NonDetectionNonScatteringStates}]
	The proof of the theorem is based on the following two results: the insertion of a second auxiliary detector (see Lemma \ref{lem:InsertionAHDetector} below) and the $L^2$-convergence of two-particle Haag--Ruelle scattering states (see Theorem \ref{thm:L2Convergence} and Proposition~\ref{prop:L2Convergence} in Section~\ref{sec:L2Convergence}).
	
	We may assume that an $\varepsilon>0$ exists such that $E(M\leq 2m+\varepsilon)\psi = 0$, where $M=\sqrt{H^2-|\mathbf{P}|^2}$ is the mass operator (i.e.~$\psi$ lies above the two-particle threshold). Otherwise, we approximate $\psi$ by such elements. Moreover, we may assume that $\Delta$ is sufficiently small. Otherwise, we decompose $\Delta=\bigcup_i \Delta_i$ into finitely many sufficiently small compact sets $\Delta_i$ and prove the theorem for $\psi_i=E(\Delta_i)\psi$. 
	
	Let $\{\hat{g}_j\}_j\subset C_c^\infty(\mathbb{R}^d)$ be a locally finite smooth partition of unity. Because the Arveson spectrum of $B$ is compact, an $\hat{f}\in C_c^\infty(\mathbb{R}^d)$ exists such that $B=B(f)=\int_{\mathbb{R}^d} f(x)B(x)\diff x$ and
	\begin{align}
		B(f) = \sum_j B(f*g_j),
	\end{align} 
	where only finitely many summands are non-zero. The Arveson spectrum of $B_j=B(f*g_j)$ is a subset of $\supp(\hat{f}\hat{g}_j)\cap \ArSp{B}$. Hence, by choosing an appropriate partition of unity, we may assume that, for every $j$, the Arveson spectrum of $B_j$ is sufficiently small. Observe that $B_j^*$ is not necessarily a creation operator because it may happen that $\sigma_\alpha(B_j^*)\cap \sigma(P)=\emptyset$. It remains to prove that, for every $j$,
	\begin{align}\label{eq:SimplifiedStatement}
		\lim_{t \to \infty} \e^{\I t H} \int_{\mathbb{R}^s} h\left(\frac{\mathbf{x}}{t}\right) (B^* B_j)(\mathbf{x}) \diff \mathbf{x} \e^{-\I t H}\psi = 0.
	\end{align}	
	Under the assumptions of the theorem, it occurs that either $\emptyset\neq\overline{\Delta-\sigma_\alpha(B^*_j)}\cap \sigma(P)\subset H_m$ or $\overline{\Delta-\sigma_\alpha(B^*_j)}\cap \sigma(P)=\emptyset$. If the latter is true, then $B_j\psi=0$ and \eqref{eq:SimplifiedStatement} is trivial. Thus, we may assume $\emptyset\neq\overline{\Delta-\sigma_\alpha(B^*_j)}\cap \sigma(P)\subset H_m$. We distinguish between the following three cases:
	\begin{enumerate}
		\item \label{i:Overlap} $\emptyset\neq\sigma_\alpha(B_j^*)\cap \sigma(P)\subset H_m$ and the sets $\pi_{\mathbf{P}}(\overline{\Delta-\sigma_\alpha(B^*_j)})$ and $\pi_{\mathbf{P}}(\sigma_\alpha(B^*_j))$ overlap.
		\item \label{i:Separated} $\emptyset\neq\sigma_\alpha(B_j^*)\cap \sigma(P)\subset H_m$ and the sets $\pi_{\mathbf{P}}(\overline{\Delta-\sigma_\alpha(B^*_j)})$ and $\pi_{\mathbf{P}}(\sigma_\alpha(B^*_j))$ are separated.
		\item \label{i:NoCreation} $\sigma_\alpha(B_j^*)\cap \sigma(P)=\emptyset$.
	\end{enumerate}	
	The list is exhaustive because $\ArSp{B_j^*}\subset \ArSp{B^*}$ and $\ArSp{B^*}\cap\sigma(P)\subset H_m$. We can exclude Case~\ref{i:Overlap} if $\Delta$ lies above the two-particle threshold and $\Delta$ as well as $\sigma_\alpha(B_j)$ are sufficiently small (this can be assumed by the arguments above). For the proof of this claim, note that two vectors on the mass shell add up to a vector above the two-particle threshold only if the two vectors are distinct. 
	
	For the remaining two cases, note that if $\tilde{\Delta}$ is the closure of $\overline{\Delta-\ArSp{B_j^*}}+\ArSp{B^*}$, then we obtain the following estimate from the Cauchy--Schwarz inequality:	
	\begin{align} \label{eq:InsertionSecondDetector}
		&\|\e^{\I tH} \int_{\mathbb{R}^s} h\left(\frac{\mathbf{x}}{t}\right) (B^*B_j)(\mathbf{x}) \diff \mathbf{x} \e^{-\I t H}\psi \| 
		= \sup_{\|\phi\|=1} |\scp{\e^{-\I t H}E(\tilde{\Delta})\phi}{\int_{\mathbb{R}^s} h\left(\frac{\mathbf{x}}{t}\right) (B^*B_j)(\mathbf{x}) \diff \mathbf{x} \ \e^{-\I t H} \psi}| \notag \\
		&\leq \|h\|_{L^\infty} \sup_{\|\phi\|=1} \left( \int_{\mathbb{R}^s} \|B(\mathbf{x})\e^{-\I t H}E(\tilde{\Delta})\phi\|^2 \diff \mathbf{x}\right)^\frac{1}{2} \left( \int_{\mathbb{R}^s} \| B_j(\mathbf{x}) \e^{-\I t H}\psi \|^2 \diff \mathbf{x} \right)^\frac{1}{2}.		
	\end{align}
	It suffices to prove that the second factor in brackets converges to 0 because the first factor in brackets is bounded by Proposition~\ref{prop:HRCreationOperatorL2Bound}. We observe that, for every $(t,\mathbf{x})\in\mathbb{R}^d$,
	\begin{align}
		B_j(\mathbf{x})\e^{-\I tH} \psi = E(\overline{\Delta-\sigma_\alpha(B^*_j)}) B_j(\mathbf{x})\e^{-\I tH} \psi.
	\end{align}	
	By Lemma \ref{lem:InsertionAHDetector}, a creation operator $C^*$ exists such that the Arveson spectrum $\ArSp{C^*}$ lies in an open neighbourhood of $\overline{\Delta-\sigma_\alpha(B^*_j)}$ and, for every $t\in\mathbb{R}$,
	\begin{align}\label{eq:DetectorInsertion}
		\int_{\mathbb{R}^s} \| B_j(\mathbf{x}) \e^{-\I t H}\psi \|^2 \diff \mathbf{x} =
		\int_{\mathbb{R}^s} \int_{\mathbb{R}^s} |\scp{\psi}{\e^{\I t H} B_j^*(\mathbf{x})C^*(\mathbf{y}) \Omega}|^2 \diff \mathbf{x} \diff \mathbf{y}.
	\end{align}	
	In Case~\ref{i:Separated}, we can choose $C^*$ such that $\pi_{\mathbf{P}}(\ArSp{C^*})$ is separated from $\pi_{\mathbf{P}}(\ArSp{B_j^*})$. The operator $\e^{-\I t (\omega(D_{\mathbf{x}})+\omega(D_{\mathbf{y}}))}$ is an isometry on $L^2(\mathbb{R}^{2s})$; hence, the r.h.s.~of \eqref{eq:DetectorInsertion} equals $\norm{\scp{\psi}{\varphi_t}}_{L^2}^2$, where
	\begin{align}
		\varphi_t(\mathbf{x},\mathbf{y}) = \e^{\I t (H-\omega(D_\mathbf{x})-\omega(D_\mathbf{y}))} B_j^*(\mathbf{x})C^*(\mathbf{y}) \Omega.
	\end{align}	
	We refer to Section~\ref{ssec:CarlemanFunctions} for the definition of $\varphi_t$ as a Carleman function. In Case~\ref{i:Separated}, $\scp{\psi}{\varphi_t}$ converges in $L^2(\mathbb{R}^{2s})$ to $\scp{\psi}{\varphi_+}$ as $t\to\infty$ by Theorem~\ref{thm:L2Convergence}, and $\scp{\psi}{\varphi_+} = 0$ because $\psi \in (\mathcal{H}^{\mathrm{out}})^\perp$. In Case~\ref{i:NoCreation}, $\scp{\psi}{\varphi_t}$ converges in $L^2(\mathbb{R}^{2s})$ to 0 as $t\to\infty$ by Proposition~\ref{prop:L2Convergence}.
\end{proof}

The following lemma, which we applied in the above proof of Theorem~\ref{thm:NonDetectionNonScatteringStates}, demonstrates that one-particle states are accessible through detectors. Specifically, for every one-particle state $\psi$, we construct an Araki--Haag detector that is triggered by this state. The main idea of the proof is to identify one-particle states $\psi\in\mathfrak{h}_m$ with wave functions in $L^2(\mathbb{R}^s)$ and to choose a creation operator $B^*$ such that the wave function of $B^*\Omega$ is 1 on a given compact set.

\begin{lem} \label{lem:InsertionAHDetector}
	Let $\Delta \subset \mathbb{R}^d$ be compact such that $\Delta\cap \sigma(P) \subset H_m$. A creation operator $B^*$ exists that satisfies
	\begin{align}\label{eq:AHD2}
		E(\Delta) = E(\Delta) \int_{\mathbb{R}^s} (B^*B)(\mathbf{x}) \diff\mathbf{x}\ E(\Delta).
	\end{align}
	For every $\varepsilon>0$, the creation operator $B^*$ can be chosen such that its Arveson spectrum $\ArSp{B^*}$ is contained in an $\varepsilon$-neighbourhood of $\Delta$.
\end{lem}

\begin{proof}
	The strategy of the proof is to demonstrate that the r.h.s.~of \eqref{eq:AHD2} is equal to the l.h.s.~of \eqref{eq:AHD2} for a suitable creation operator $B^*$. It suffices to demonstrate that, for all $\phi,\psi\in E(\Delta)\mathcal{H}$,
	\begin{align}
		\scp{\phi}{\psi} = \int_{\mathbb{R}^s} \scp{\phi}{(B^*B)(\mathbf{x})\psi} \diff\mathbf{x}.
	\end{align}
	Because states in $E(\Delta)\mathcal{H}$ are one-particle states, it holds that, for every creation operator $B^*$,
	\begin{align}\label{eq:VacuumProjection}
		B(\mathbf{x})\psi = \scp{\Omega}{B(\mathbf{x})\psi}\Omega,
	\end{align}
	and similarly for $\phi$. The set of generalised momentum eigenvectors $\{|\mathbf{p}\rangle\}_{\mathbf{p}\in\mathbb{R}^s}$ introduced in Section~\ref{ssec:ArakiHaagFormula} obeys a completeness relation in the one-particle space: $\int_{\mathbb{R}^s} |\mathbf{p}\rangle \langle \mathbf{p}| \diff\mathbf{p}$ is the identity in $\mathfrak{h}_m$. We utilise \eqref{eq:VacuumProjection} and this completeness relation to obtain the following identity:
	\begin{align}
		\int_{\mathbb{R}^s} \scp{\phi}{(B^*B)(\mathbf{x})\psi} \diff\mathbf{x} = \int_{\mathbb{R}^{s}} \scp{\phi}{B^*(\mathbf{x})\Omega} \scp{B^*(\mathbf{x})\Omega}{\psi} \diff \mathbf{x}
		= \int_{\mathbb{R}^{s}} \diracscp{\phi}{\mathbf{p}} \diracscp{\mathbf{p}}{\psi} |\diracscp{\mathbf{p}}{B^*\Omega}|^2 \diff\mathbf{p}.
	\end{align}
	The distributions $\mathbf{p} \mapsto \diracscp{\phi}{\mathbf{p}}$ and $\mathbf{p} \mapsto \diracscp{\mathbf{p}}{\psi}$ are compactly supported with support contained in $\pi_\mathbf{P}(\Delta)$. Moreover, a creation operator $B^*$ exists such that $\mathbf{p} \mapsto |\diracscp{\mathbf{p}}{B^*\Omega}|^2$ is smooth, $|\diracscp{\mathbf{p}}{B^*\Omega}|^2 = 1$ on a given compact set, and $|\diracscp{\mathbf{p}}{B^*\Omega}|^2 = 0$ on a slightly larger set \cite[Section~5.3~(a)]{araki1999}. It follows that, for every $\varepsilon>0$, we can choose a creation operator $B^*$ in such a way that
	\begin{align}
		\int_{\mathbb{R}^{s}} \diracscp{\phi}{\mathbf{p}} \diracscp{\mathbf{p}}{\psi} |\diracscp{\mathbf{p}}{B^*\Omega}|^2 \diff\mathbf{p}
		= \int_{\mathbb{R}^{s}} \diracscp{\phi}{\mathbf{p}} \diracscp{\mathbf{p}}{\psi} \diff\mathbf{p} = \scp{\phi}{\psi}
	\end{align}
	and $\ArSp{B^*}$ is contained in an $\varepsilon$-neighbourhood of $\Delta$.
\end{proof}

\section{$L^2$-convergence of two-particle scattering states} \label{sec:L2Convergence}

In this section, we state and prove a new convergence result for Haag--Ruelle scattering states that we applied in the proof of Theorem~\ref{thm:NonDetectionNonScatteringStates}. We fix two creation operators $B_1^*, B_2^*$. In the following, we use the same notation as in Section~\ref{ssec:CarlemanFunctions} for $\varphi_t$ (with $n=2$), that is,
\begin{align}
	\varphi_t(\mathbf{x},\mathbf{y}) = \e^{\I t (H-\omega(D_{\mathbf{x}})-\omega(D_{\mathbf{y}}))} B_1^*(\mathbf{x}) B_2^*(\mathbf{y}) \Omega.
\end{align}
Remember that, for every $\psi\in\mathcal{H}$, $\scp{\psi}{\varphi_t} \in L^2(\mathbb{R}^{2s})$ and $\scp{\psi}{\varphi_t}$ converges pointwise to $\scp{\psi}{\varphi_+}\in L^2(\mathbb{R}^{2s})$. We improve pointwise convergence to convergence in $L^2(\mathbb{R}^{2s})$.
\begin{thm}\label{thm:L2Convergence}
	Let $B_1^*$, $B_2^*$ be two creation operators such that the momentum transfers $\pi_{\mathbf{P}}(\sigma_\alpha(B_1^*))$, $\pi_{\mathbf{P}}(\sigma_\alpha(B_2^*))$ are separated. For every $\psi\in \mathcal{H}_{\mathrm{ac}}(P)$, $\scp{\psi}{\varphi_t}$ converges in $L^2(\mathbb{R}^{2s})$ to $\scp{\psi}{\varphi_+}$ as $t\to\infty$.
\end{thm}

\begin{proof}
	\begin{enumerate}[wide, labelwidth=!, labelindent=0pt, label=(\roman*)]
		\item It suffices to prove the theorem for vectors $\psi$ from a dense subset $\mathcal{D}\subset \mathcal{H}_{\mathrm{ac}}(P)$. In fact, a vector $\psi \in \mathcal{H}_{\mathrm{ac}}(P)$ is approximated by a sequence $(\psi_n)_{n\in\mathbb{N}}$ in $\mathcal{D}$; hence,
		\begin{align} \label{eq:DensityArgument}
			\norm{\scp{\psi}{\varphi_t-\varphi_+}}_{L^2} \leq 2\norm{\psi-\psi_n}_{\mathcal{H}} \norm{\varphi_0}_\mathrm{C} + \norm{\scp{\psi_n}{\varphi_t-\varphi_+}}_{L^2}.
		\end{align}
		On the r.h.s.~of \eqref{eq:DensityArgument}, we take the limit $t\to\infty$ and, subsequently, the limit $n\to \infty$.  In the following, we choose $\mathcal{D}=\mathcal{M}(P)$ (see Definition~\ref{defn:RadonNikodymDerivativeBounded}), that is, we assume that the Radon--Nikodym derivative $\rho_\psi$ of the spectral measure $\scp{\psi}{E(\cdot)\psi}$ is a bounded function. The space $\mathcal{M}(P)$ is dense in $\mathcal{H}_{\mathrm{ac}}(P)$ by Lemma~\ref{lem:MSpaceDense}. We denote by $\normm{\psi}$ the $L^\infty$-norm of $\sqrt{\rho_\psi}$.
		
		\item \label{itm:CooksMethod} We prove by Cook's method that, for $\psi\in\mathcal{M}(P)$, $\scp{\psi}{\varphi_t}$ is a Cauchy sequence in $L^2(\mathbb{R}^{2s})$:
		\begin{align}\label{eq:CooksMethod}
			\|\scp{\psi}{\varphi_{t_2}-\varphi_{t_1}}\|_{L^2} = \left\|\int_{t_1}^{t_2} \partial_\tau \scp{\psi}{\varphi_\tau} \diff \tau \right\|_{L^2}.
		\end{align}
		We claim that the time derivative amounts to replacing the product of $B_1^*(\mathbf{x})$ and $B_2^*(\mathbf{y})$ by a commutator of two creation operators. In fact,
		\begin{align}
			\partial_\tau\varphi_\tau(\mathbf{x},\mathbf{y}) = \I \e^{\I \tau (H-\omega(D_\mathbf{x})-\omega(D_\mathbf{y}))} (H-\omega(D_\mathbf{x})-\omega(D_\mathbf{y})) B_1^*(\mathbf{x})B_2^*(\mathbf{y})\Omega,
		\end{align}
		and, because $\omega(D_\mathbf{y})B_2^*(\mathbf{y})\Omega = \omega(\mathbf{P}) B_2^*(\mathbf{y})\Omega = H B_2^*(\mathbf{y})\Omega$, we obtain
		\begin{align}
			(H-\omega(D_\mathbf{x})-\omega(D_\mathbf{y})) B_1^*(\mathbf{x})B_2^*(\mathbf{y})\Omega
			= [H,B_1^*](\mathbf{x})B_2^*(\mathbf{y})\Omega - \omega(D_\mathbf{x})B_1^*(\mathbf{x})B_2^*(\mathbf{y})\Omega.
		\end{align}
		If we write $B_1^*=B_1^*(f)$ for a Schwartz function $f\in\mathcal{S}(\mathbb{R}^d)$ that satisfies $\hat{f}=1$ on $\ArSp{B_1^*}$, then
		\begin{align}
			[H,B_1^*](\mathbf{x}) &= B_1^*(-D_0 f)(\mathbf{x}), \\
			\omega(D_\mathbf{x}) B_1^*(\mathbf{x}) &= B_1^*(\omega(\mathbf{D})f)(\mathbf{x}).
		\end{align}
		We conclude that
		\begin{align}
			\partial_\tau\varphi_\tau(\mathbf{x},\mathbf{y}) &= \I\e^{\I \tau (H-\omega(D_\mathbf{x})-\omega(D_\mathbf{y}))} \tilde{B}_1^*(\mathbf{x})B_2^*(\mathbf{y})\Omega \notag \\
			&= \I\e^{\I \tau (H-\omega(D_\mathbf{x})-\omega(D_\mathbf{y}))} [\tilde{B}_1^*(\mathbf{x}),B_2^*(\mathbf{y})]\Omega,
		\end{align}
		where $\tilde{B}_1^* = B_1^*(g)$ and $g\in\mathcal{S}(\mathbb{R}^d)$ is any function such that $\hat{g}(p) = p_0-\omega(\mathbf{p})$ on $\sigma_\alpha(B_1^*)$. To obtain the commutator, we used $\tilde{B}_1^*\Omega = (H-\omega(\mathbf{P}))B_1^*\Omega = 0$. 
		
		\item The sets $K_{\mathrm{tot}}=\pi_{\mathbf{P}}(\sigma_\alpha(B_1^*))+\pi_{\mathbf{P}}(\sigma_\alpha(B_2^*))$ and $K_{\mathrm{rel}}=\pi_{\mathbf{P}}(\sigma_\alpha(B_1^*))-\pi_{\mathbf{P}}(\sigma_\alpha(B_2^*))$ contain the total and relative momentum support of the function $(\mathbf{x},\mathbf{y}) \mapsto [\tilde{B}_1^*(\mathbf{x}),B_2^*(\mathbf{y})]\Omega$, respectively, that is,
		\begin{align} \label{eq:TotRelMomentum}
			[\tilde{B}_1^*(\mathbf{x}),B_2^*(\mathbf{y})]\Omega = \chi(D_{\mathbf{x}}+D_{\mathbf{y}} \in K_{\mathrm{tot}}) \chi(D_{\mathbf{x}}-D_{\mathbf{y}} \in K_{\mathrm{rel}}) [\tilde{B}_1^*(\mathbf{x}),B_2^*(\mathbf{y})]\Omega,
		\end{align}
		where $\chi$ denotes the characteristic function. By assumption on the Arveson spectra of $B_1^*$ and $B_2^*$, the sets $K_{\mathrm{tot}}$ and $K_\mathrm{rel}$ are compact and $K_{\mathrm{rel}}$ is separated from 0.		
	
		\item It is convenient to introduce relative coordinates:
		\begin{align}
			\mathbf{u} &= \mathbf{x} - \mathbf{y}, &D_{\mathbf{u}} &= \frac{1}{2}(D_{\mathbf{x}}-D_{\mathbf{y}}), \\
			\mathbf{v} &= \frac{1}{2}(\mathbf{x} + \mathbf{y}), &D_{\mathbf{v}} &= D_{\mathbf{x}}+D_{\mathbf{y}}.
		\end{align}
		If we formulate \eqref{eq:CooksMethod} in relative coordinates, we arrive at the following identity:
		\begin{align}
			\| \scp{\psi}{\varphi_{t_2}-\varphi_{t_1}} \|_{L^2}^2
			= \int_{\mathbb{R}^{2s}} \left| \int_{t_1}^{t_2} \e^{-\I \tau (\omega(\frac{1}{2}D_\mathbf{v}+D_{\mathbf{u}}) + \omega(\frac{1}{2}D_\mathbf{v}-D_{\mathbf{u}}))} \scp{\psi}{\e^{\I\tau H} \e^{-\I\mathbf{v}\cdot\mathbf{P}}\phi(\mathbf{u})} \diff \tau \right|^2 \diff\mathbf{u} \diff\mathbf{v},
		\end{align}
		where
		\begin{align}
			\phi(\mathbf{u}) = \e^{-\frac{\I}{2}\mathbf{u}\cdot\mathbf{P}} [\tilde{B}_1^*,B_2^*(-\mathbf{u})]\Omega
		\end{align}
		is a Hilbert space-valued Schwartz function. The function $\phi$ is smooth because it has bounded energy-momentum, and $\phi$ decays rapidly because the commutator $[\tilde{B}_1^*,B_2^*(-\mathbf{u})]$ decays rapidly in norm by Lemma~\ref{lem:AlmostLocalCommutatorDecay}.	Moreover, by \eqref{eq:TotRelMomentum},
		\begin{align} \label{eq:TotRelMomentum2}
			\scp{\psi}{\e^{\I\tau H} \e^{-\I\mathbf{v}\cdot\mathbf{P}}\phi(\mathbf{u})} 
			= \chi(D_{\mathbf{v}}\in K_{\mathrm{tot}}) \chi(2D_{\mathbf{u}}\in K_{\mathrm{rel}}) \scp{\psi}{\e^{\I\tau H} \e^{-\I\mathbf{v}\cdot\mathbf{P}}\phi(\mathbf{u})}.
		\end{align}
	
		\item We apply Plancherel's theorem in the $\mathbf{v}$-integral (we denote the Fourier transformation by $\mathcal{F}$), and duality in the $\mathbf{u}$-integral:
		\begin{align} \label{eq:P1}
			&\int_{\mathbb{R}^{s}} \int_{\mathbb{R}^{s}} \left| \int_{t_1}^{t_2} \e^{-\I \tau (\omega(\frac{1}{2}\mathbf{p}+D_{\mathbf{u}})+\omega(\frac{1}{2}\mathbf{p}-D_{\mathbf{u}}))} \mathcal{F}_{\mathbf{v}\to\mathbf{p}}\scp{\psi}{\e^{\I\tau H} \e^{-\I\mathbf{v}\cdot\mathbf{P}}\phi(\mathbf{u})} \diff \tau \right|^2 \diff\mathbf{u} \diff\mathbf{p} \notag \\
			&=\int_{\mathbb{R}^{s}} \sup_{\|f\|_{L^2}=1} \left| \int_{t_1}^{t_2} \int_{\mathbb{R}^{s}} \overline{f(\mathbf{u})} \e^{-\I \tau \omega_\mathbf{p}(D_{\mathbf{u}})} \mathcal{F}_{\mathbf{v}\to\mathbf{p}}\scp{\psi}{\e^{\I\tau H} \e^{-\I\mathbf{v}\cdot\mathbf{P}}\phi(\mathbf{u})} \diff\mathbf{u} \diff \tau \right|^2 \diff\mathbf{p},
		\end{align}
		where we introduced, for $\mathbf{p}\in\mathbb{R}^s$, the following operator on $L^2(\mathbb{R}^s)$:
		\begin{align} \label{eq:EnergyTwoParticlesFixedTotalMomentum}
			\omega_{\mathbf{p}}(D_{\mathbf{u}}) = \omega(\mathbf{p}/2+D_{\mathbf{u}}) + \omega(\mathbf{p}/2-D_{\mathbf{u}}).
		\end{align}	
		The operator $\omega_{\mathbf{p}}(D_{\mathbf{u}})$ corresponds to the energy of two free particles with relativistic dispersion relation and total momentum $\mathbf{p}$. Observe that the minimal value of the Fourier multiplier $\mathbf{q}\mapsto \omega_\mathbf{p}(\mathbf{q}) = \omega(\mathbf{p}/2+\mathbf{q}) + \omega(\mathbf{p}/2-\mathbf{q})$ is $2\omega(\mathbf{p}/2)$. This value is assumed if and only if $\mathbf{q}=0$ (i.e.~if and only if the relative momentum is 0). 
	
		\item The area of integration of the $\mathbf{p}$-integral in \eqref{eq:P1} can be restricted to the total momentum support $K_{\mathrm{tot}}$ due to \eqref{eq:TotRelMomentum2}. Moreover, because the relative momentum support $K_{\mathrm{rel}}$ is separated from 0, an $\varepsilon>0$ exists such that, for all $\mathbf{p} \in K_{\mathrm{tot}}$,
		\begin{align}
			E_\mathbf{p}(I_{\mathbf{p},\varepsilon}) \phi(\mathbf{u}) = \phi(\mathbf{u}),
		\end{align}	
		where $E_\mathbf{p}$ is the spectral measure of $\omega_\mathbf{p}(D_\mathbf{u})$, $I_{\mathbf{p},\varepsilon} = [2\omega(\mathbf{p}/2)+\varepsilon,\beta]$, and
		\begin{align} \label{eq:beta}
			\beta = \sup_{\mathbf{p}\in K_{\mathrm{tot}}} \sup_{\mathbf{q}\in K_{\mathrm{rel}}} \omega_{\mathbf{p}}(\mathbf{q}) = \sup_{\mathbf{p}\in\pi_{\mathbf{P}}(\sigma_\alpha(B_1^*))} \sup_{\mathbf{q}\in\pi_{\mathbf{P}}(\sigma_\alpha(B_2^*))} (\omega(\mathbf{p}) + \omega(\mathbf{q})) < \infty
		\end{align}
		is the maximal energy of two free relativistic particles with total momentum in $K_{\mathrm{tot}}$ and relative momentum in $K_\mathrm{rel}$.
	
		\item Let $\theta \in C_c^\infty(0,\infty)$ satisfy $\theta=1$ on $(\varepsilon/2,\beta+1)$. For $\mathbf{p}\in\mathbb{R}^s$, set $\theta_\mathbf{p}(\lambda)=\theta(\lambda-2\omega_{\mathbf{p}}(\mathbf{p}/2))$ and
		\begin{align}
			\mathbf{F}_\mathbf{p}(\mathbf{q}) = \theta_\mathbf{p}(\omega_\mathbf{p}(\mathbf{q})) \frac{\nabla\omega_{\mathbf{p}}(\mathbf{q})}{|\nabla\omega_{\mathbf{p}}(\mathbf{q})|^2},
		\end{align}
		where 
		\begin{align} \label{eq:GradientOmegaP}
			\nabla\omega_{\mathbf{p}}(\mathbf{q}) =  \frac{\frac{1}{2}\mathbf{p}+\mathbf{q}}{\omega(\frac{1}{2}\mathbf{p}+\mathbf{q})} -\frac{\frac{1}{2}\mathbf{p}-\mathbf{q}}{\omega(\frac{1}{2}\mathbf{p}-\mathbf{q})}.
		\end{align}
		The function $\mathbf{F}_\mathbf{p}(\mathbf{q})$ is well-defined because $|\nabla\omega_{\mathbf{p}}(\mathbf{q})|\geq b>0$ for all $\mathbf{q}\in\omega_{\mathbf{p}}^{-1}(\supp(\theta_\mathbf{p}))$. This follows from the fact that $\nabla\omega_{\mathbf{p}}(\mathbf{q})=0$ if and only if $\mathbf{q}=0$ and 0 is separated from the set $\omega_{\mathbf{p}}^{-1}(\supp(\theta_\mathbf{p}))$. The optimal value of $b$ depends on $\mathbf{p}\in \mathbb{R}^s$, but $b>0$ can be chosen independent of $\mathbf{p}$ as long as $\mathbf{p}$ ranges over compact subsets of $\mathbb{R}^s$. We define the following modified dilation operator in relative coordinates:
		\begin{align} \label{eq:ConjugateOperator}
			A_{\mathbf{p}} = \frac{1}{2} (\mathbf{F}_\mathbf{p}(D_\mathbf{u})\cdot\mathbf{u} + \mathbf{u}\cdot \mathbf{F}_\mathbf{p}(D_\mathbf{u})) = \mathbf{F}_\mathbf{p}(D_\mathbf{u})\cdot \mathbf{u} + \frac{\I}{2} (\div \mathbf{F}_\mathbf{p})(D_{\mathbf{u}}),
		\end{align}
		where $\mathbf{u}$ is the multiplication operator by the relative coordinate $\mathbf{u}$. The operator $A_{\mathbf{p}}$ is essentially self-adjoint on the Schwartz space $\mathcal{S}(\mathbb{R}^s)$ \cite[Lemma~7.6.4]{amrein1996}. We denote its self-adjoint closure by the same symbol. The divergence of $\mathbf{F}_\mathbf{p}$ in \eqref{eq:ConjugateOperator} can be computed explicitly:
		\begin{align}
			\div\mathbf{F}_\mathbf{p}(\mathbf{q}) = \theta_\mathbf{p}'(\omega_{\mathbf{p}}(\mathbf{q})) - \theta_\mathbf{p}(\omega_{\mathbf{p}}(\mathbf{q})) \frac{\Delta \omega_{\mathbf{p}}(\mathbf{q})}{|\nabla\omega_\mathbf{p}(\mathbf{q})|^2},
		\end{align}
		where the Laplacian of $\omega_{\mathbf{p}}$ is a bounded function:
		\begin{align} \label{eq:LaplacianOmegaP}
			\Delta\omega_{\mathbf{p}}(\mathbf{q}) &= \frac{s}{\omega(\frac{1}{2}\mathbf{p}+\mathbf{q})} 
			+ \frac{s}{\omega(\frac{1}{2}\mathbf{p}-\mathbf{q})} 
			-\frac{|\frac{1}{2}\mathbf{p}+\mathbf{q}|^2}{\omega(\frac{1}{2}\mathbf{p}+\mathbf{q})^3} 
			-\frac{|\frac{1}{2}\mathbf{p}-\mathbf{q}|^2}{\omega(\frac{1}{2}\mathbf{p}-\mathbf{q})^3}.
		\end{align}			
		We take the operator $A_\mathbf{p}$ as a conjugate operator for $\omega_{\mathbf{p}}(D_\mathbf{u})$ because the commutator
		\begin{align}\label{eq:Commutator}
			[\omega_\mathbf{p}(D_\mathbf{u}),\I A_\mathbf{p}] = \theta_\mathbf{p}(\omega_\mathbf{p}(D_\mathbf{u}))
		\end{align}
		assumes a simple form. Formally, we obtain \eqref{eq:Commutator} by utilising the well-known commutation relation $[\omega_{\mathbf{p}}(D_\mathbf{u}),\I\mathbf{u}] = \nabla\omega_{\mathbf{p}}(D_\mathbf{u})$; for more details, we refer to the comments below the proof and the example subsequent to Definition~\ref{defn:MourreEstimate}. We denote by $\jap{\cdot}$ the multiplication operator on $L^2(\mathbb{R}^s)$ that maps $f$ to $\mathbf{u}\mapsto\jap{\mathbf{u}}f(\mathbf{u})$. From \eqref{eq:ConjugateOperator}, the boundedness of $\Delta\omega_{\mathbf{p}}$, and $|\nabla\omega_{\mathbf{p}}(\mathbf{q})|\geq b>0$ for $\mathbf{q}\in\omega_{\mathbf{p}}^{-1}(\supp(\theta_\mathbf{p}))$, it follows that, for all $f\in \mathcal{S}(\mathbb{R}^s)$, $\|\jap{A_{\mathbf{p}}}\jap{\cdot}^{-1}f\| \leq \|\jap{\cdot}^{-1}f\| + \|A_\mathbf{p}\jap{\cdot}^{-1}f\| \leq C\|f\|$, where $C$ can be chosen independent of $\mathbf{p}$ as long as $\mathbf{p}$ ranges over compact subsets of $\mathbb{R}^s$. Thus, $\jap{A_\mathbf{p}}\jap{\cdot}^{-1}$ extends to a bounded operator, which is bounded by $C$. From interpolation (Lemma~\ref{lem:Interpolation} with $X=1$), it follows that $\|\jap{A_{\mathbf{p}}}^\nu\jap{\cdot}^{-\nu}\|\leq C^\nu$ and $\|\jap{A_{\mathbf{p}}}^\nu f\|\leq C^\nu \|\jap{\cdot}^\nu f\|$ for $\nu\in[0,1]$.
		
		\item Let $\nu\in(1/2,1]$. We insert $1=\jap{A_\mathbf{p}}^{-\nu} \jap{A_\mathbf{p}}^\nu$ into \eqref{eq:P1} and apply the Cauchy--Schwarz inequality to arrive at the following bound of \eqref{eq:P1}:
		\begin{align} \label{eq:CSBound}
			&\int_{K_{\mathrm{tot}}} \left( \sup_{\|f\|_{L^2}=1} \int_{t_1}^{t_2} \|\jap{A_{\mathbf{p}}}^{-\nu} \e^{\I \tau \omega_\mathbf{p}(D_{\mathbf{u}})} E_\mathbf{p}(I_{\mathbf{p},\varepsilon}) f\|^2_{L^2} \diff\tau \right) \notag \\
			&\times \left(\int_{t_1}^{t_2} \|\jap{A_{\mathbf{p}}}^\nu \mathcal{F}_{\mathbf{v}\to\mathbf{p}}\scp{\psi}{\e^{\I\tau H} \e^{-\I\mathbf{v}\cdot\mathbf{P}}\phi} \|^2_{L^2} \diff \tau \right) \diff \mathbf{p}.
		\end{align}	
		By Lemma \ref{lem:ApSmooth} below, for every $\mathbf{p}\in\mathbb{R}^s$, a constant $c(\mathbf{p})<\infty$ exists such that
		\begin{align}\label{eq:LocalDecayEstimate}
			\sup_{\|f\|_{L^2}=1} \int_{-\infty}^\infty \|\jap{A_{\mathbf{p}}}^{-\nu} \e^{\I \tau \omega_\mathbf{p}(D_{\mathbf{u}})} E_\mathbf{p}(I_{\mathbf{p},\varepsilon}) f\|^2_{L^2} \diff\tau \leq c(\mathbf{p}),
		\end{align}
		and $\sup_{\mathbf{p}\in K_{\mathrm{tot}}} c(\mathbf{p})<\infty$ because $K_{\mathrm{tot}}$ is compact. It remains to prove that 
		\begin{align}
			&\int_{K_\mathrm{tot}} \int_{t_1}^{t_2} \|\jap{A_{\mathbf{p}}}^\nu \mathcal{F}_{\mathbf{v}\to\mathbf{p}}\scp{\psi}{\e^{\I\tau H} \e^{-\I\mathbf{v}\cdot\mathbf{P}}\phi} \|^2_{L^2} \diff \tau \diff \mathbf{p} \notag \\
			&\leq C^{2\nu} \int_{\mathbb{R}^s} \int_{t_1}^{t_2} \|\jap{\cdot}^\nu \mathcal{F}_{\mathbf{v}\to\mathbf{p}}\scp{\psi}{\e^{\I\tau H} \e^{-\I\mathbf{v}\cdot\mathbf{P}}\phi} \|^2_{L^2} \diff \tau \diff \mathbf{p} 
		\end{align}
		converges to 0 as $t_1,t_2\to\infty$. For this, it suffices to observe that, by Proposition~\ref{prop:SquareIntegrability},
		\begin{align} \label{eq:IntegralMicrocausality}
			\int_{-\infty}^\infty \int_{\mathbb{R}^{s}} \int_{\mathbb{R}^{s}} |\jap{\mathbf{u}}^\nu\scp{\psi}{\e^{\I\tau H} \e^{-\I \mathbf{v}\cdot\mathbf{P}} \phi(\mathbf{u})}|^2 \diff\mathbf{u} \diff\mathbf{v} \diff \tau
			\leq (2\pi)^d \normm{\psi}^2 \int_{\mathbb{R}^s} \|\jap{\mathbf{u}}^\nu \phi(\mathbf{u})\|^2 \diff\mathbf{u} < \infty.
		\end{align}
		Note that the last integral is finite because $\|\phi(\mathbf{u})\|$ decays rapidly. \qedhere
	\end{enumerate}
\end{proof}

To complete the proof of Theorem~\ref{thm:L2Convergence}, it remains to demonstrate \eqref{eq:LocalDecayEstimate}. This bound can be derived using Mourre's conjugate operator method (see Section~\ref{ssec:LocallySmoothOperators}). For a fixed $\mathbf{p}\in\mathbb{R}^s$, \eqref{eq:LocalDecayEstimate} directly follows from a Mourre estimate for $\omega_{\mathbf{p}}(D_\mathbf{u})$. Specifically, $\omega_{\mathbf{p}}(D_\mathbf{u})$ obeys a Mourre estimate with conjugate operator $A_\mathbf{p}$ on the open set $J_{\mathbf{p},\varepsilon}=(2\omega({\mathbf{p}}/2)+\varepsilon/2,\beta+1)$, which contains the compact interval $I_{\mathbf{p},\varepsilon}$. To prove this claim, we refer to the example subsequent to Definition~\ref{defn:MourreEstimate}, that is, we must confirm the estimates \eqref{eq:GradientEstimates} for $h=\omega_\mathbf{p}$. We computed the gradient and Laplacian of $\omega_\mathbf{p}$ in \eqref{eq:GradientOmegaP} and \eqref{eq:LaplacianOmegaP}, respectively. The gradient $\nabla\omega_{\mathbf{p}}$ is bounded from below by a positive constant $b>0$ for all $\mathbf{q}\in\omega_{\mathbf{p}}^{-1}(J_{\mathbf{p},\varepsilon})$, and $\Delta\omega_{\mathbf{p}}$ is a bounded function; thus, $|\Delta\omega_{\mathbf{p}}(\mathbf{q})|\leq \|\Delta\omega_{\mathbf{p}}\|_\infty b^{-2}|\nabla\omega_{\mathbf{p}}(\mathbf{q})|^2$. This implies $\omega_\mathbf{p}(D_\mathbf{u})\in C^\infty(A_\mathbf{p})$ and the following Mourre estimate:
\begin{align}\label{eq:MourreEstimateOmegaP}
	E_\mathbf{p}(J_{\mathbf{p},\varepsilon})[\omega_\mathbf{p}(D_\mathbf{u}),\I A_\mathbf{p}]E_\mathbf{p}(J_{\mathbf{p},\varepsilon})=\theta_\mathbf{p}(\omega_\mathbf{p}(D_\mathbf{u}))E_\mathbf{p}(J_{\mathbf{p},\varepsilon})=E_\mathbf{p}(J_{\mathbf{p},\varepsilon}).
\end{align}
The function $\theta$ was chosen such that $\theta_\mathbf{p}=1$ on $J_{\mathbf{p},\varepsilon}$. By Proposition~\ref{prop:OptimalConstant1}, the Mourre estimate \eqref{eq:MourreEstimateOmegaP} yields \eqref{eq:LocalDecayEstimate} with
\begin{align}\label{eq:cp}
	c(\mathbf{p}) = 8\sup_{\lambda\in I_{\mathbf{p},\varepsilon},\mu\in(0,1)} \|\jap{A_\mathbf{p}}^{-\nu}\Im(\omega_\mathbf{p}(D_\mathbf{u})-\lambda-\I\mu)^{-1}\jap{A_{\mathbf{p}}}^{-\nu}\| < \infty,
\end{align}
where $c(\mathbf{p})$ is finite for every $\mathbf{p}\in\mathbb{R}^s$ due to the limiting absorption principle (Theorem~\ref{thm:LimitingAbsorptionPrinciple}). However, determining the dependence of $c(\mathbf{p})$ on the total momentum $\mathbf{p}$ seems to be non-trivial (i.e.~verifying $\sup_{\mathbf{p}\in K_{\mathrm{tot}}} c(\mathbf{p})<\infty$). In the proof of the following lemma, we demonstrate that $c(\mathbf{p})$ is bounded by a function that is continuous in $\mathbf{p}$.

\begin{lem}\label{lem:ApSmooth}
	Let $\mathbf{p}\in\mathbb{R}^s$ and $\varepsilon>0$. Let $\omega_{\mathbf{p}}(D_\mathbf{u})$ and $A_\mathbf{p}$ be the operators defined in \eqref{eq:EnergyTwoParticlesFixedTotalMomentum} and \eqref{eq:ConjugateOperator}, respectively, and set $I_{\mathbf{p},\varepsilon}=[2\omega(\mathbf{p}/2)+\varepsilon,\beta]$ for a $\beta\in(2\omega(\mathbf{p}/2)+\varepsilon,\infty)$. For every $\nu>1/2$, a constant $c(\mathbf{p})$ exists such that, for $f\in L^2(\mathbb{R}^s)$,
	\begin{align} \label{eq:LocalDecayEst}
		\int_{-\infty}^\infty \|\jap{A_{\mathbf{p}}}^{-\nu} \e^{\I \tau \omega_\mathbf{p}(D_{\mathbf{u}})} E_\mathbf{p}(I_{\mathbf{p},\varepsilon}) f\|^2_{L^2} \diff\tau \leq c(\mathbf{p}) \|f\|_{L^2}^2,
	\end{align}
	where $\sup_{\mathbf{p}\in K} c(\mathbf{p})<\infty$ for every compact set $K\subset\mathbb{R}^s$.
\end{lem}

\begin{proof}
	We may assume $1/2<\nu\leq1$ because if \eqref{eq:LocalDecayEst} holds for a given $\nu=\nu_0$, then it obviously holds for all $\nu\geq\nu_0$. As argued above, it suffices to demonstrate that $c(\mathbf{p})$ as chosen in \eqref{eq:cp} is bounded by a continuous function. To obtain such a bound, we apply Proposition~\ref{prop:LimitingAbsorptionPrincipleSpectralGap} with $\lambda_0=0$, where we must verify a limiting absorption principle for the inverse operator $\omega_{\mathbf{p}}(D_\mathbf{u})^{-1}$. The inverse exists and is a bounded operator because $\omega_\mathbf{p}(D_\mathbf{u}) \geq 2m>0$ is bounded from below by a positive constant. 
		
	Let $\tilde{E}_\mathbf{p}$ be the spectral measure of $\omega_\mathbf{p}(D_\mathbf{u})^{-1}$. The Mourre estimate \eqref{eq:MourreEstimateOmegaP} for $\omega_\mathbf{p}(D_\mathbf{u})$ on $J_{\mathbf{p},\varepsilon}$ implies a Mourre estimate for the inverse operator $\omega_{\mathbf{p}}(D_\mathbf{u})^{-1}$ on $\tilde{J}_{\mathbf{p},\varepsilon} = \{\lambda \in \mathbb{R} \mid \lambda^{-1} \in J_{\mathbf{p},\varepsilon}\}$:
	\begin{align}
		-\tilde{E}_\mathbf{p}(\tilde{J}_{\mathbf{p},\varepsilon}) [\omega_\mathbf{p}(D_\mathbf{u})^{-1},\I A_{\mathbf{p}}] \tilde{E}_\mathbf{p}(\tilde{J}_{\mathbf{p},\varepsilon})
		&=\tilde{E}_\mathbf{p}(\tilde{J}_{\mathbf{p},\varepsilon}) \omega_{\mathbf{p}}(D_\mathbf{u})^{-1} [\omega_\mathbf{p}(D_\mathbf{u}),\I A_{\mathbf{p}}] \omega_\mathbf{p}(D_\mathbf{u})^{-1} \tilde{E}_\mathbf{p}(\tilde{J}_{\mathbf{p},\varepsilon}) \notag \\
		&\geq \frac{1}{(1+\beta)^2} \tilde{E}_\mathbf{p}(\tilde{J}_{\mathbf{p},\varepsilon}),
	\end{align}	
	where we applied \eqref{eq:CommutatorResolvent} in the first step. We remark that $A_\mathbf{p}$ depends implicitly on $\varepsilon$ through $\theta$ and that $A_\mathbf{p}$ is only well-defined as long as $\varepsilon>0$. Let $\tilde{I}_{\mathbf{p},\varepsilon}=\{\lambda \in \mathbb{R} \mid \lambda^{-1} \in I_{\mathbf{p},\varepsilon}\} \subset\tilde{J}_{\mathbf{p},\varepsilon}$. From Proposition~\ref{prop:LimitingAbsorptionPrincipleSpectralGap}, we obtain the following estimate for $c(\mathbf{p})$:
	\begin{align}
		\frac{1}{8}c(\mathbf{p}) &\leq \sup_{\lambda\in I_{\mathbf{p},\varepsilon},\mu\in(0,1)} \|\jap{A_{\mathbf{p}}}^{-\nu} (\omega_\mathbf{p}(D_\mathbf{u})-\lambda-\I\mu)^{-1} \jap{A_{\mathbf{p}}}^{-\nu}\| \notag \\
		&\leq\sup_{\lambda\in \tilde{I}_{\mathbf{p},\varepsilon},\mu>0} |\lambda| \left( |\lambda|+\frac{1}{|\lambda|} + \|\jap{A_{\mathbf{p}}}^{-\nu} (\omega_{\mathbf{p}}(D_\mathbf{u})^{-1}-\lambda+\I\mu)^{-1}\jap{A_{\mathbf{p}}}^{-\nu}\| \right) \notag \\ &\ \ \ \ \ \ \ \ \ \times \|\jap{A_{\mathbf{p}}}^\nu \omega_{\mathbf{p}}(D_\mathbf{u})^{-1} \jap{A_{\mathbf{p}}}^{-\nu}\|.
	\end{align}
	The contributions of the supremum from $\mu\geq1$ are uncritical; hence, we can restrict $\mu$ to $(0,1)$. In the remaining expression, we bound the factors individually:  For $\lambda \in \tilde{I}_{\mathbf{p},\varepsilon}$, we have $|\lambda|\leq (2\omega_\mathbf{p}(\mathbf{p}/2)+\varepsilon)^{-1}$ and $|\lambda|^{-1}\leq \beta$. From Theorem~\ref{thm:LimitingAbsorptionPrincipleBoundedOperator}, we obtain the following estimate: 
	\begin{align}
		\sup_{\lambda\in \tilde{I}_{\mathbf{p},\varepsilon}, \mu\in(0,1)} \|\jap{A_{\mathbf{p}}}^{-\nu} (\omega_{\mathbf{p}}(D_\mathbf{u})^{-1}-\lambda\mp\I\mu)^{-1}\jap{A_{\mathbf{p}}}^{-\nu}\| \notag \\
		\leq \left[ \left( \frac{4}{a\varepsilon_0} + \frac{c_3 \varepsilon_0^{\nu}}{\nu} \right)^{\frac{1}{2}} + \frac{c_3\varepsilon_0^{\nu-\frac{1}{2}}}{\nu-\frac{1}{2}} \right]^2 \e^{c_3 \varepsilon_0},
	\end{align}
	where
	\begin{align}
		c_1\equiv c_1(\mathbf{p}) &= \|\omega_{\mathbf{p}}(D_\mathbf{u})^{-1}\| + \|[\omega_{\mathbf{p}}(D_\mathbf{u})^{-1},A_\mathbf{p}]\| + (1+4(1+\beta)^2)\|[\omega_{\mathbf{p}}(D_\mathbf{u})^{-1},A_\mathbf{p}]\|^2 \notag \\
		&+ \|[[\omega_{\mathbf{p}}(D_\mathbf{u})^{-1},A_\mathbf{p}],A_\mathbf{p}]\| + \frac{1}{(1+\beta)^2} + \frac{1}{\beta} + \frac{1}{|2\omega(\mathbf{p}/2)+\varepsilon|} + 1, \\
		c_2\equiv c_2(\mathbf{p}) &= \left(\sqrt{2}+\frac{\sqrt{8c_1(\mathbf{p})}}{\delta(\mathbf{p})}\right)(1+\beta), \\
		c_3\equiv c_3(\mathbf{p}) &= 4c_2(\mathbf{p}) + 2c_1(\mathbf{p})c_2(\mathbf{p})^2, \\
		\varepsilon_0\equiv \varepsilon_0(\mathbf{p}) &= \min\left\{\frac{\delta(\mathbf{p})}{4(1+\beta)c_1(\mathbf{p})}, \frac{\delta(\mathbf{p})^2}{16c_1(\mathbf{p})^2}\right\} ,
	\end{align}
	and $\delta\equiv\delta(\mathbf{p})>0$ is a continuous function such that $\tilde{I}_{\mathbf{p},\varepsilon}+\delta(\mathbf{p}) \subset J_{\mathbf{p},\varepsilon}$. Moreover, by interpolation (Lemma~\ref{lem:Interpolation}) and \eqref{eq:ResolventConjugateOperatorRelation}, it holds that, for $\nu\in[0,1]$,
	\begin{align}
		&\|\jap{A_{\mathbf{p}}}^\nu \omega_{\mathbf{p}}(D_\mathbf{u})^{-1} \jap{A_{\mathbf{p}}}^{-\nu}\| 
		\leq \|\omega_{\mathbf{p}}(D_\mathbf{u})^{-1}\|^{1-\nu} \|\jap{A_{\mathbf{p}}} \omega_{\mathbf{p}}(D_\mathbf{u})^{-1} \jap{A_{\mathbf{p}}}^{-1}\|^{\nu}
		\notag \\ 
		&\leq \|\omega_{\mathbf{p}}(D_\mathbf{u})^{-1}\|^{1-\nu}(\|\omega_{\mathbf{p}}(D_\mathbf{u})^{-1}\|+\|[A_{\mathbf{p}}, \omega_{\mathbf{p}}(D_\mathbf{u})^{-1}](A_{\mathbf{p}}+\I)^{-1}\|)^\nu < \infty,
	\end{align}
	where we used $\|\jap{A_\mathbf{p}}(A_\mathbf{p}+\I)^{-1}\|=\|(A_\mathbf{p}+\I)\jap{A_\mathbf{p}}^{-1}\| = 1$ in the second step. We observe that all expressions depend continuously on $\mathbf{p}$; hence, $c(\mathbf{p})$ is bounded by a continuous function. 	
\end{proof}

In the following proposition, we prove the statement of Theorem~\ref{thm:L2Convergence} under slightly modified assumptions. Specifically, we drop the requirement for the momentum transfers of $B_1^*$ and $B_2^*$ to be separated, but we assume $B_1^*\Omega=0$. This assumption simplifies the proof because we obtain the commutator of the two almost local operators $B_1^*,B_2^*$ from $B_1^*\Omega=0$ (see Step~\ref{itm:CooksMethod} in the proof of Theorem~\ref{thm:L2Convergence}).

\begin{prop}\label{prop:L2Convergence}
	If $B_1^*, B_2^*$ are almost local operators with compact Arveson spectrum such that $B_1^*\Omega=0$, then, for every $\psi\in\mathcal{H}_{\mathrm{ac}}(P)$,
	\begin{align}
		\lim_{t \to \infty} \int_{\mathbb{R}^{s}} \int_{\mathbb{R}^{s}} |\scp{\psi}{\e^{\I tH} B_1^*(\mathbf{x})B_2^*(\mathbf{y})\Omega}|^2 \diff\mathbf{x}\diff\mathbf{y} = 0.
	\end{align}
\end{prop}

\begin{proof}
	As in the proof of Theorem~\ref{thm:L2Convergence}, it suffices to prove the proposition for $\psi\in\mathcal{M}(P)$. We obtain the following identity from the assumption $B_1^*\Omega=0$:
	\begin{align}
		\int_{\mathbb{R}^{2s}} |\scp{\psi}{\e^{\I tH} B_1^*(\mathbf{x})B_2^*(\mathbf{y})\Omega}|^2 \diff\mathbf{x}\diff\mathbf{y}
		= \int_{\mathbb{R}^{2s}} |\scp{\psi}{\e^{\I tH} \e^{-\I\mathbf{x}\cdot\mathbf{P}} [B_1^*,B_2^*(\mathbf{y})]\Omega}|^2 \diff\mathbf{x}\diff\mathbf{y};
	\end{align}
	hence, by Proposition~\ref{prop:SquareIntegrability},
	\begin{align}
		\int_{-\infty}^\infty \left( \int_{\mathbb{R}^{2s}} |\scp{\psi}{\e^{\I tH} B_1^*(\mathbf{x})B_2^*(\mathbf{y})\Omega}|^2 \diff\mathbf{x}\diff\mathbf{y} \right) \diff t \leq (2\pi)^d \normm{\psi}^2 \int_{\mathbb{R}^s} \|[B_1^*,B_2^*(\mathbf{y})]\Omega\|^2 \diff \mathbf{y} < \infty.
	\end{align}
	The last integral is finite by Lemma~\ref{lem:AlmostLocalCommutatorDecay} because $B_1^*$ and $B_2^*$ are almost local. This implies that the $L^2(\mathbb{R}^{2s})$-valued function $g(t;\mathbf{x},\mathbf{y})=\scp{\psi}{\e^{\I tH} B_1^*(\mathbf{x})B_2^*(\mathbf{y})\Omega}$ is B-convergent to 0 \cite[Definition~4.1.1]{arendt2011}, that is, for every $\delta>0$, 
	\begin{align}
		\lim_{t\to\infty} \frac{1}{\delta} \int_{t}^{t+\delta} g(\tau) \diff\tau = 0.
	\end{align}
	In fact, the B-convergence of $g$ is a consequence of the following estimate:
	\begin{align}
		\norm{\frac{1}{\delta} \int_{t}^{t+\delta} g(\tau) \diff\tau}_{L^2} \leq \frac{1}{\sqrt{\delta}} \left(\int_{t}^{t+\delta} \|g(\tau)\|^2_{L^2} \diff\tau \right)^{\frac{1}{2}} \stackrel{t\to\infty}{\longrightarrow} 0.
	\end{align}	
	To prove that $g$ also converges in $L^2(\mathbb{R}^{2s})$, we demonstrate that $g$ is slowly oscillating \cite[Definition~4.2.1]{arendt2011}, that is, for every $\varepsilon>0$, a $\delta>0$ and $t_0\geq0$ exist such that $\|g(t)-g(s)\|\leq\varepsilon$ whenever $s,t\geq t_0$ and $|t-s|\leq\delta$. If $g$ is slowly oscillating and B-convergent to 0, then we deduce from a simple Tauberian theorem \cite[Theorem~4.2.3]{arendt2011} that $g$ converges in $L^2(\mathbb{R}^{2s})$ to 0. In order to estimate
	\begin{align}
		\|g(t)-g(s)\|^2 = \int_{\mathbb{R}^{2s}} |\scp{\psi}{\e^{-\I\mathbf{x}\cdot\mathbf{P}} (\e^{\I tH}-\e^{\I sH}) [B_1^*,B_2^*(\mathbf{y})]\Omega}|^2 \diff\mathbf{x}\diff\mathbf{y},
	\end{align}
	we apply Proposition~\ref{prop:SquareIntegrability} to the family $(\mathbf{P},H)$. We select $a=\{1,\dots,s\}\subset \mathcal{N} = \{1,\dots,s+1\}$ so that $(\mathbf{P},H)_a=\mathbf{P}$. The vector $[B_1^*,B_2^*(\mathbf{y})]\Omega$ has bounded energy, that is, for a $p_0\in\mathbb{R}$,
	\begin{align} \label{eq:BoundedEnergy}
		[B_1^*,B_2^*(\mathbf{y})]\Omega = E(H\leq p_0)[B_1^*,B_2^*(\mathbf{y})]\Omega.
	\end{align}
	Taking \eqref{eq:BoundedEnergy} into account, \eqref{eq:TfHSmoothBoundedEnergy} of Proposition~\ref{prop:SquareIntegrability} yields the following estimate:
	\begin{align}
		\|g(t)-g(s)\|^2
		&\leq (2\pi)^s 2p_0 \normm{\psi}^2 \int_{\mathbb{R}^s} \|(\e^{\I tH}-\e^{\I sH})[B_1^*,B_2^*(\mathbf{y})]\Omega\|^2 \diff \mathbf{y} \notag \\
		&\leq (2\pi)^s 2p_0^3 \normm{\psi}^2 |t-s|^2 \int_{\mathbb{R}^s} \|[B_1^*,B_2^*(\mathbf{y})]\Omega\|^2 \diff \mathbf{y};
	\end{align}
	hence, $g$ is slowly oscillating.
\end{proof}

\section{Applications and outlook} \label{sec:Outlook}

In this section, we discuss the relevance of Theorem~\ref{thm:MainResult} to the problem of asymptotic completeness in quantum field theory, its applicability to models, and its extension to spin systems. Furthermore, we outline a strategy for proving the convergence of Araki--Haag detectors in regions of the multi-particle spectrum above the three-particle threshold.

\subsection{Asymptotic completeness} \label{ssec:AsymptoticCompleteness}

Our main result has potential implications for a proof of asymptotic completeness in local relativistic quantum field theory, which is a long-standing open problem, as discussed in Section~1. Theorem~\ref{thm:MainResult} alone does not imply asymptotic completeness because it also applies to models that are not asymptotically complete, such as certain generalised free fields. To bridge the gap from Theorem~\ref{thm:MainResult} to asymptotic completeness, an additional condition is necessary. We presented one such condition in Corollary~\ref{cor:AsymptoticCompleteness}. Another potential condition, which is easier to verify in models, could be that the Hamiltonian $H$ can be written as a space integral over a local energy density \cite[p.~278]{haag1996}, represented schematically as
\begin{align}
	H = \int_{\mathbb{R}^s} T^{00}(\mathbf{x}) \diff \mathbf{x},
\end{align}
where $T^{00}(\mathbf{x})$ is the 00-component of the energy-momentum tensor (see Condition~T in \cite{dybalski2010} for an appropriate smearing of the energy-momentum tensor). For a non-zero state $\psi\in\mathcal{H}$ orthogonal to the vacuum vector, we then have
\begin{align}
	0 < \scp{\psi}{H\psi} = \int_{\mathbb{R}^s} \scp{\e^{-\I tH}\psi}{T^{00}(\mathbf{x}) \e^{-\I tH}\psi} \diff \mathbf{x}.
\end{align}
We expect that the r.h.s.~converges to 0 as $t\to\infty$ if $\psi$ is not a scattering state, similar to Theorem~\ref{thm:MainResult}, where $T^{00}(\mathbf{x})$ functions as the detector. This would provide a contradiction to $\scp{\psi}{H\psi}>0$, implying that all states are scattering states. However, $T^{00}(\mathbf{x})$ is not in the canonical form $(B^*B)(\mathbf{x})$ analysed in this paper, where $B^*$ is a creation operator. Notably, it has been shown that in the free massive scalar field theory, the energy-momentum tensor restricted to subspaces of bounded energy can be approximated by a sum of operators of the form $B^*B$, where $B$ is almost local and energy-decreasing \cite[Section~D]{dybalski2008}. This result would make Theorem~\ref{thm:MainResult} applicable if it could be established that $B^*$ is a creation operator. Extending the results of \cite[Section~D]{dybalski2008} to interacting models is an interesting direction towards proving asymptotic completeness in quantum field theory.

\subsection{Models -- Free products of Borchers triples} \label{ssec:Models}

Essential for applying Theorem~\ref{thm:MainResult} to models is the assumption for the multi-particle spectrum to be absolutely continuous in the two-particle region. An interesting class of models, which meets this requirement, emerges from the free product constructions of Borchers triples by Longo, Tanimoto, and Ueda \cite{longo2019}. A two-dimensional Borchers triple $\mathcal{B} = (\mathcal{M},U,\Omega)$ comprises a von Neumann algebra $\mathcal{M}$, a unitary representation $U$ of $\mathbb{R}^2$ with joint spectrum in $V_+$, and a cyclic and separating vector $\Omega$ such that $\Omega$ is invariant under $U$ and $U(x)\mathcal{M}U(x)^*\subset \mathcal{M}$ for all $x\in W_R$, where $W_R$ is the right wedge. Typical examples of Borchers triples stem from Haag--Kastler nets with $\mathcal{M}$ being the wedge algebra generated by the local observable algebras $\mathcal{R}(O)$, $O\subset W_R$. Conversely, given a Borchers triple $(\mathcal{M}, U, \Omega)$, we can construct a local net by setting $\mathcal{R}(D_{a,b}) = U(a)\mathcal{M}U(a)^* \cap U(b)\mathcal{M}' U(b)^*$. Here, $D_{a,b} = (W_R+a)\cap (W_L+b)$, $a,b\in\mathbb{R}^2$, is a double cone, and $W_L$ is the left wedge. The resulting net satisfies microcausality, isotony, and Poincaré covariance \cite[Section~2.1.2]{longo2019}.

Consider two identical copies $\mathcal{B}_1, \mathcal{B}_2$ of the Borchers triple corresponding to the two-dimensional free massive scalar field theory \cite[Section~5.1]{longo2019}. According to \cite[Proposition~5.1]{longo2019}, the free product $\mathcal{B}_1 \star \mathcal{B}_2$ again forms a Borchers triple $\mathcal{B}=(\mathcal{M},U,\Omega)$. Because $U$ is the free product of $U_1$ and $U_2$ (i.e., essentially, a direct sum involving only the unitary representations $U_1$ and $U_2$ of the free theory), the multi-particle spectrum of $U$ is absolutely continuous. Furthermore, the two-particle $S$-matrix of the free product model is non-trivial, yet asymptotic completeness fails \cite[Section~5.3]{longo2019}. 

We remark that the non-triviality of the local algebras arising from the free product construction was not proved in \cite{longo2019}. If it could be established that the vacuum vector is cyclic for the local algebras, then the free product construction would present an interesting class of models where the convergence result of Theorem~\ref{thm:MainResult} was not known previously.

\subsection{Spin systems}

Theorem~\ref{thm:MainResult} can be extended to spin systems through the adapted Haag--Ruelle scattering theory developed by Bachmann, Dybalski, and Naaijkens \cite{bachmann2016}. A spin system is a $C^*$-dynamical system $(\mathcal{A},\tau)$, where $\tau$ is a unitary representation of the spacetime translation group. The algebra $\mathcal{A}$ is generated by a local net $\{\mathcal{A}(\Lambda)\}_\Lambda$, where $\Lambda \subset \mathbb{Z}^s$ is a bounded spatial region. The challenge in adapting the Haag--Ruelle scattering theory to spin systems is to define a suitable almost local algebra. By replacing the double cones $K_r$ in Definition~\ref{defn:AlmostLocal} with open balls in $\mathbb{Z}^s$ of radius $r$, we obtain an algebra that is a priori not invariant under time translations. However, by utilising the Lieb--Robinson bound,
\begin{align}
	\|[\tau_t(A),B]\| \leq C_{A,B} \e^{\lambda (v_{\mathrm{LR}}t-d(A,B))},
\end{align}
where $\lambda>0$ is a constant, $v_{LR}>0$ the Lieb--Robinson velocity, and $d(A,B)$ the distance between the localisation regions of the local observables $A$ and $B$, it is possible to define an almost local algebra that is invariant under spacetime translations and satisfies Lemma~\ref{lem:AlmostLocalCommutatorDecay} \cite[Theorem~3.10]{bachmann2016}.

The proof of Theorem~\ref{thm:MainResult} readily extends to spin systems with only minor adjustments, provided that the isolated mass shell in the energy-momentum spectrum is regular and pseudo-relativistic. These properties of the mass shell have been verified for the Ising model in a strong magnetic field \cite[Section~6]{bachmann2016}.

To our knowledge, it is unknown whether the energy-momentum spectrum of the Ising model is absolutely continuous in the two-particle region. Thus far, the spectral analysis of spin systems has focused on perturbation theory of the energy spectrum \cite{pokorny1992,yarotsky2006,nachtergaele2023,delVecchio2023}. Nevertheless, the cited papers showcase a diverse array of techniques available for investigating spectral properties of spin systems, which could potentially be extended to analyse the energy-momentum spectrum as well.

Analogous to quantum field theory, the adapted version of Theorem~\ref{thm:MainResult} could be utilised to prove two-particle asymptotic completeness in spin systems. Notably, Buchholz \cite{buchholz1986} constructed an ideal local detector $C$ under assumptions which are typical for spin systems. If it were possible to verify that this detector has the canonical form $C=B^*B$, where $B^*$ is a creation operator, it would provide a proof of two-particle asymptotic completeness through Corollary~\ref{cor:AsymptoticCompleteness}.

\subsection{Convergence of Araki--Haag detectors in the many-particle region}

Extending the convergence of Araki--Haag detectors into regions of the multi-particle spectrum above the three-particle threshold presents an interesting direction for further research. Relevant for this problem is the $L^2$-convergence of many-body Haag--Ruelle scattering states:
\begin{align}
	\varphi_t(\mathbf{x}_1,\dots,\mathbf{x}_n) = \e^{\I t (H-\omega(D_{\mathbf{x}_1})-\cdots-\omega(D_{\mathbf{x}_n}))} B_1^*(\mathbf{x}_1)\dots B_n^*(\mathbf{x}_n)\Omega.
\end{align}
In Section~\ref{sec:L2Convergence}, we established the convergence of $\scp{\psi}{\varphi_t}$, $\psi\in\mathcal{H}_\mathrm{ac}(P)$, in $L^2(\mathbb{R}^{ns})$ for the case $n=2$, but extending this proof to $n\geq 3$ requires new ideas. The main difficulty lies in Step~\ref{itm:CooksMethod} of the proof, where differentiation w.r.t.~$t$ yields commutators of creation operators. While for $n=2$ the resulting expression decays in the relative coordinate $\mathbf{u}$, we obtain anisotropic decay for $n\geq 3$ (e.g.~the expression does not decay if two or more particles stay close to each other while the remaining particles escape to infinity). This anisotropic decay resembles the behaviour encountered in the scattering theory of many-body quantum mechanical systems, where the many-body potential exhibits similar anisotropic decay. Physically, the anisotropic decay of the potential is associated with the formation of many-body bound states, which must be considered to describe the asymptotic evolution of the quantum system correctly.

A promising approach to adapt the proof of Theorem~\ref{thm:L2Convergence} to the many-body case is to transfer Yafaev's proof of asymptotic completeness \cite{yafaev1993} to quantum field theory. Yafaev's technique is a natural extension of Lavine's theorem to the many-body case. His proof relies on certain radiation estimates, which trace back to Sommerfeld's radiation condition. In the context of quantum field theory, such radiation estimates could prove equally relevant.

\subsection*{Acknowledgements}

I thank Wojciech Dybalski for continuous support during the work on this project, countless helpful discussions, and assistance preparing the publication. I thank Wojciech also for providing the proof of Lemma~\ref{lem:InsertionAHDetector}. Moreover, I thank Paweł Duch for carefully proofreading different drafts of this paper. The research leading to these results received funding from National Science Centre, Poland, under the grant `Sonata Bis' 2019/34/E/ST1/00053.

\appendix

\section{Mourre's conjugate operator method} \label{appx:ConjugateOperatorMethod}

Mourre's conjugate operator method is a powerful tool to analyse spectral properties of a self-adjoint operator $H:D(H)\to\mathcal{H}$ on a Hilbert space $\mathcal{H}$ based on a strictly positive commutator estimate (the so-called Mourre estimate, see Section~\ref{ssec:MourreEstimate}). One of the main results of the conjugate operator method is the limiting absorption principle (see Section~\ref{ssec:LAP}), which controls the resolvent $R(z)=(H-z)^{-1}$ as the resolvent parameter $z\in\rho(H)$ approaches the spectrum. Closely related to the limiting absorption principle are locally smooth operators, which we discuss in Section~\ref{ssec:LocallySmoothOperators}.

\subsection{Mourre estimate} \label{ssec:MourreEstimate}

We introduce the following regularity classes, which are relevant for defining commutators.

\begin{defn}
	Let $A$ be a self-adjoint operator on $\mathcal{H}$ and let $k\in\mathbb{N}\cup\{\infty\}$. We denote by $C^k(A)$ the space of all self-adjoint operators $H$ such that, for a $z\in\rho(H)$, $t\mapsto \e^{\I t A}R(z)\e^{-\I t A}$ is a $C^k$-map in the strong operator topology.
\end{defn}

By \cite[Lemma~6.2.1]{amrein1996}, if $z\in \rho(H)$ exists such that $t\mapsto \e^{\I t A}R(z)\e^{-\I t A}$ is a $C^k$-map, then $t\mapsto \e^{\I t A}R(z)\e^{-\I t A}$ is a $C^k$-map for all $z\in\rho(H)$. Moreover, if $H$ is bounded, $t\mapsto \e^{\I t A}R(z)\e^{-\I t A}$ is a $C^k$-map if and only if $t\mapsto \e^{\I t A}H\e^{-\I t A}$ is a $C^k$-map. If $H\in C^1(A)$ is bounded, we can define the commutator $[H,A]$ on $\mathcal{H}$ as the strong derivative of $t\mapsto \I \e^{\I t A}H\e^{-\I t A}$. If $H\in C^1(A)$ is unbounded, the sesquilinear form defined by $HA-AH$ on $D(A)\cap D(H)$ extends to $D(H)$. 
\begin{prop}[{\cite[Theorem~6.2.10~(b)]{amrein1996}}]
	If $H\in C^1(A)$, then $D(A)\cap D(H)$ is a core for~$H$, and the sesquilinear form 
	\begin{align}
		(f,g) \mapsto \scp{Hf}{Ag} - \scp{Af}{Hg}, \ f,g \in D(A)\cap D(H),
	\end{align}
	has a unique extension to a continuous sesquilinear form on $D(H)$, where $D(H)$ is equipped with the graph topology. If we denote the operator associated to the extended sesquilinear form by $[H,A]$, then the following identity holds on $\mathcal{H}$ in the form sense:
	\begin{align}\label{eq:CommutatorResolvent}
		[R(z),A] = -R(z)[H,A]R(z), \ z \in \rho(H).
	\end{align}
\end{prop}

Observe that $[R(z),A]$ is a bounded operator on $\mathcal{H}$ if $H\in C^1(A)$ and that $R(z)$ maps $\mathcal{H}$ into $D(H)$ and the dual space $D(H)'$ into $\mathcal{H}$. We are now able to formalise the notion of a strictly positive commutator estimate. The example following the definition is applied in Section~\ref{sec:L2Convergence}.

\begin{defn} \label{defn:MourreEstimate}
	The operator $H$ obeys a \textbf{Mourre estimate} on an open bounded set $J\subset\mathbb{R}$ if a self-adjoint operator $A$ (\textbf{conjugate operator}) exists such that $H\in C^1(A)$ and, for an $a>0$,
	\begin{align}
		E(J) [H, \I A] E(J) \geq a E(J),
	\end{align}
	where $E$ is the spectral measure of $H$.
\end{defn}

\begin{exmp}[{\cite[Lemma~7.6.4]{amrein1996}}]
	Let $\mathcal{H}=L^2(\mathbb{R}^n)$, and define $H=h(D)$, where $h:\mathbb{R}^n\to \mathbb{R}$ is a Borel function and $D=-\I\partial$. Suppose $J\subset \mathbb{R}$ is an open set such that $\Omega=h^{-1}(J)$ is also open in $\mathbb{R}^n$, and $h$ belongs to $C^2$ on a neighbourhood of the closure of $\Omega$. Assume that a constant $c>0$ exists such that, for $x\in\Omega$,
	\begin{align}\label{eq:GradientEstimates}
		|\nabla h(x)| \geq c, \ |\Delta h(x)| \leq c^{-1}|\nabla h(x)|^2.
	\end{align}
	Consider $\theta\in C_c^\infty(J)$ to be real-valued, and let $F$ be defined as $F(x) = \theta(h(x))|\nabla h(x)|^{-2}\nabla h(x)$ for $x\in \Omega$, and $F(x)=0$ otherwise. The modified dilation operator
	\begin{align}
		A=\frac{1}{2}(F(D)\cdot X+X\cdot F(D))
	\end{align}
	is essentially self-adjoint on $\mathcal{S}(\mathbb{R}^n)$ and $H\in C^\infty(A)$. Moreover, for all $k\geq 1$, $\mathrm{ad}^k_{-\I A}(H) = \theta_{k-1}(H)$ are bounded operators on $\mathcal{H}$, where $\theta_k(\lambda)=[\theta(\lambda)\partial_\lambda]^k\theta(\lambda)$. In particular, 
	\begin{align}
		E(J_\theta)[H,\I A] E(J_\theta) = \theta(H)E(J_\theta) = E(J_\theta),
	\end{align}
	where $J_\theta=\{\lambda\in\mathbb{R}\mid \theta(\lambda)=1\}$ (i.e.~$H$ obeys a Mourre estimate on every open subset of $J_\theta$). 
\end{exmp}

\subsection{Limiting absorption principle} \label{ssec:LAP}

The limiting absorption principle extends, for $\nu>1/2$, the holomorphic resolvent function $\mathbb{C}_\pm \ni z \mapsto \jap{A}^{-\nu}R(z)\jap{A}^{-\nu}$ to a continuous function on $\mathbb{C}_\pm\cup J$ if $H$ obeys a Mourre estimate on $J$ with conjugate operator $A$. The limiting absorption principle has been proved under different assumptions. To our knowledge, the optimal assumptions are those in \cite[Theorem~7.4.1]{amrein1996} if $H$ has a spectral gap (i.e.~$\sigma(H)\neq\mathbb{R}$) and \cite[Theorem~0.1]{sahbani1997} if $H$ does not have a spectral gap. We prefer to cite the limiting absorption principle under less optimal assumptions, which are sufficient for our purposes and avoid introducing Besov spaces associated to a $C_0$-group.

\begin{thm}[Limiting absorption principle, {\cite[Theorem~1]{gerard2008}}] \label{thm:LimitingAbsorptionPrinciple}
	Let $H\in C^2(A)$. If $H$ obeys a Mourre estimate on $J$, then, for every compact interval $I\subset J$ and every $\nu>1/2$,
	\begin{align}\label{eq:LAP1}
		\sup_{\lambda\in I, \mu>0} \| \jap{A}^{-\nu} (H-\lambda\mp\I\mu)^{-1} \jap{A}^{-\nu} \| < \infty.
	\end{align}
\end{thm}

In our application, it is relevant to determine the dependence of the bound \eqref{eq:LAP1} on the parameter $a$, the operators $A$, $H$, and the sets $I$, $J$. It is easier to determine this dependence if $H$ is a bounded operator.  

\begin{thm}\label{thm:LimitingAbsorptionPrincipleBoundedOperator}
	Let $H\in C^2(A)$ be a bounded operator. If $H$ obeys a Mourre estimate on $J$, then, for every compact interval $I=[\alpha,\beta]\subset J$ such that $I+[-\delta,\delta] \subset J$ for a $\delta>0$, and every $\nu>1/2$,
	\begin{align}\label{eq:LAP2}
		\sup_{\lambda\in I,\mu\in(0,1)} \| \jap{A}^{-\nu} (H-\lambda\mp\I\mu)^{-1} \jap{A}^{-\nu} \| \leq \left[ \left( \frac{4}{a\varepsilon_0} + \frac{c_3 \varepsilon_0^{\nu}}{\nu} \right)^{\frac{1}{2}} + \frac{c_3\varepsilon_0^{\nu-\frac{1}{2}}}{\nu-\frac{1}{2}} \right]^2 \e^{c_3 \varepsilon_0},
	\end{align}
	where
	\begin{align}
		c_1 &= \|H\| + \|[H,A]\| + \left(1+\frac{4}{a}\right)\|[H,A]\|^2 + \|[[H,A],A]\| + a + |\alpha| + |\beta| + 1, \\
		c_2 &= \sqrt{\frac{2}{a}} + \frac{1}{\delta} \sqrt{\frac{8c_1}{a}}, \\
		c_3 &= 4c_2 + 2c_1c_2^2, \\
		\varepsilon_0 &= \min\left\{\frac{\sqrt{a}\delta}{4c_1}, \frac{\delta^2}{16c_1^2}\right\}.
	\end{align}
\end{thm}

\begin{proof}
	The proof relies on a clever approximation $G_\varepsilon^\pm(\lambda,\mu)$ of the resolvent $R(\lambda\pm\I\mu)$ such that, formally, $G_\varepsilon^\pm(\lambda,\mu)\to R(\lambda\pm\I\mu)$ as $\varepsilon\downarrow 0$. A differential inequality ensures that $\jap{A}^{-\nu} G_\varepsilon^\pm(\lambda,\mu) \jap{A}^{-\nu}$ remains bounded as $\mu\downarrow 0$ and $\varepsilon\downarrow 0$. A complete proof of the theorem is provided in \cite[Theorem~6.3]{amrein2009}. We sketch the idea of the proof to obtain the bound \eqref{eq:LAP2}. For $\lambda\in I$ and $\mu\in(0,1)$, we define the operator-valued function
	\begin{align}
		\Phi_\varepsilon = \jap{\varepsilon A}^{-1} \jap{A}^{-\nu} G_\varepsilon^\pm(\lambda,\mu) \jap{A}^{-\nu} \jap{\varepsilon A}^{-1},
	\end{align}
	where $G_\varepsilon^\pm(\lambda,\mu)$ is the inverse of $H-\lambda\mp\I\mu\mp\I \varepsilon[H,\I A]$. For $\varepsilon\in(0,\varepsilon_0)$, the existence of the inverse is a consequence of the Mourre estimate. We must show that $\Phi_\varepsilon$ remains bounded as $\varepsilon\downarrow0$ and $\mu\downarrow0$. This is achieved by proving the following differential inequality \cite[(6.56)]{amrein2009}:
	\begin{align}
		\| \Phi_\varepsilon' \| \leq c_3\varepsilon^{\nu-1} + c_3 \varepsilon^{\nu-\frac{3}{2}} \|\Phi_\varepsilon\|^{\frac{1}{2}} + c_3 \|\Phi_\varepsilon\|.
	\end{align}
	Moreover, $\|\Phi_{\varepsilon_0}\| \leq 4/(a\varepsilon_0)$ \cite[p.~279]{amrein2009}. We complete the proof by Lemma~\ref{lem:MethodDiffInequality} below.
\end{proof}

\begin{lem}[Method of differential inequality, {\cite[Section~6.2.1]{amrein2009}}] \label{lem:MethodDiffInequality}
	Let $\varepsilon_0>0$ and $(0,\varepsilon_0) \ni \varepsilon \mapsto \Phi_\varepsilon$ be a continuously differentiable $\mathfrak{B}(\mathcal{H})$-valued function. If $\Phi_\varepsilon$ is a solution of the differential inequality
	\begin{align}
		\|\Phi_\varepsilon'\| \leq \theta_1(\varepsilon) + \theta_2(\varepsilon) \|\Phi_\varepsilon\|^\frac{1}{2} + \gamma \|\Phi_\varepsilon\|,
	\end{align}
	where $\theta_k:(0,\varepsilon_0)\to [0,\infty)$ satisfy $\int_0^{\varepsilon_0} \theta_k(\varepsilon) \diff \varepsilon < \infty$, $k\in\{1,2\}$, then $\Phi_\varepsilon$ is bounded as follows:
	\begin{align}
		\|\Phi_\varepsilon\| \leq \left[ \left(\|\Phi_{\varepsilon_0}\| + \int_{0}^{\varepsilon_0} \theta_1(\varepsilon') \diff \varepsilon' \right)^\frac{1}{2} + \int_0^{\varepsilon_0} \theta_2(\varepsilon') \diff \varepsilon' \right]^2 \e^{\gamma \varepsilon_0}.
	\end{align}
\end{lem}
 
If $H$ has a spectral gap, it is straightforward to obtain an explicit estimate for \eqref{eq:LAP1} from Theorem~\ref{thm:LimitingAbsorptionPrincipleBoundedOperator}.

\begin{prop} \label{prop:LimitingAbsorptionPrincipleSpectralGap}
	Let $H\in C^2(A)$ have a spectral gap. Select $\lambda_0\in\mathbb{R}\backslash \sigma(H)$ and define $R=(H-\lambda_0)^{-1}$. If $H$ obeys a Mourre estimate on $J\subset\sigma(H)$, then, for every compact interval $I\subset J$ and every $1/2<\nu\leq 1$,
	\begin{align}\label{eq:LAPSpectralGap}
		&\sup_{\lambda\in I, \mu\in(0,1)} \| \jap{A}^{-\nu} (H-\lambda\mp\I\mu)^{-1} \jap{A}^{-\nu} \| \notag \\
		&\leq \sup_{\lambda\in \tilde{I},\mu>0} |\lambda| \left( |\lambda| + \frac{1}{|\lambda|} + \|\jap{A}^{-\nu} (R-\lambda\pm\I\mu)^{-1}\jap{A}^{-\nu}\| \right) \|\jap{A}^\nu R\jap{A}^{-\nu}\| < \infty,
	\end{align}
	where $\tilde{I}=\{(\lambda-\lambda_0)^{-1} \mid \lambda \in I\}$.
\end{prop}

\begin{proof}
	The main steps of the proof are the same as in the proof of \cite[Theorem~7.4.1]{amrein1996}. The resolvents of $H$ and $R$ are related as follows:
	\begin{align}
		(H-\lambda\mp\I\mu)^{-1} = -(\lambda-\lambda_0\pm\I\mu)^{-1} [R-(\lambda-\lambda_0\pm\I\mu)^{-1}]^{-1} R.
	\end{align}
	This identity entails the following estimate:
	\begin{align}
		&\sup_{\lambda\in I, \mu\in(0,1)} \| \jap{A}^{-\nu} (H-\lambda\mp\I\mu)^{-1} \jap{A}^{-\nu} \| \notag \\
		&\leq \sup_{\lambda\in I, \mu\in(0,1)} |\lambda-\lambda_0\pm\I\mu|^{-1} \|\jap{A}^{-\nu} (R-(\lambda-\lambda_0\pm\I\mu)^{-1})^{-1}\jap{A}^{-\nu}\| \|\jap{A}^\nu R\jap{A}^{-\nu}\|.
	\end{align}
	For $z\in\rho(R)$, define $Q(z)=(R-z)^{-1}$. From the resolvent formula, we obtain the following identity:
	\begin{align}
		&Q((\lambda-\lambda_0\pm\I\mu)^{-1}) - Q\left( (\lambda-\lambda_0)^{-1} \mp \I\frac{\mu}{(\lambda-\lambda_0)^2+\mu^2} \right) \notag \\
		&=\frac{-\mu^2}{[(\lambda-\lambda_0)^2+\mu^2](\lambda-\lambda_0)} Q((\lambda-\lambda_0\pm\I\mu)^{-1})Q\left( (\lambda-\lambda_0)^{-1} \mp \I\frac{\mu}{(\lambda-\lambda_0)^2+\mu^2} \right).
	\end{align}
	Remember that $\|Q(z)\| \leq |\Im(z)|^{-1}$; hence,
	\begin{align}
		\left\|Q((\lambda-\lambda_0\pm\I\mu)^{-1}) - Q\left( (\lambda-\lambda_0)^{-1} \mp \I\frac{\mu}{(\lambda-\lambda_0)^2+\mu^2} \right)\right\| \leq \frac{(\lambda-\lambda_0)^2+\mu^2}{|\lambda-\lambda_0|},
	\end{align}
	and, subsequently,
	\begin{align}
		&\sup_{\lambda\in I, \mu\in(0,1)} |\lambda-\lambda_0\pm\I\mu|^{-1} \|\jap{A}^{-\nu} (R-(\lambda-\lambda_0\pm\I\mu)^{-1})^{-1}\jap{A}^{-\nu}\| \notag \\
		&\leq \sup_{\lambda\in I,\mu>0} \frac{1}{|\lambda-\lambda_0|} \left( \frac{(\lambda-\lambda_0)^2+1}{|\lambda-\lambda_0|} + \|\jap{A}^{-\nu} (R-(\lambda-\lambda_0)^{-1}\pm\I\mu)^{-1}\jap{A}^{-\nu}\| \right).
	\end{align}	
	It remains to demonstrate that the r.h.s.~of \eqref{eq:LAPSpectralGap} is finite. If $K\subset J$ is a compact subset, then $R$ obeys a Mourre estimate on every open subset contained in $\tilde{K}=\{(\lambda-\lambda_0)^{-1} \mid \lambda \in K\}$ (see \cite[Proposition~7.2.5]{amrein1996}); hence, by Theorem~\ref{thm:LimitingAbsorptionPrincipleBoundedOperator},
	\begin{align}
		\sup_{\lambda\in \tilde{I},\mu>0} \|\jap{A}^{-\nu} (R-\lambda\pm\I\mu)^{-1}\jap{A}^{-\nu}\| < \infty.
	\end{align}
	Also, $\|\jap{A}^{\nu} R \jap{A}^{-\nu}\|<\infty$ for $\nu\in[0,1]$ because $R\in C^1(A)$. In fact, 
	\begin{align}\label{eq:ResolventConjugateOperatorRelation}
		(A+\I)R(A+\I)^{-1} = R + [A,R](A+\I)^{-1}
	\end{align}
	is a sum of bounded operators. It follows that $\jap{A}R\jap{A}^{-1}$ is also bounded, and from interpolation (Lemma~\ref{lem:Interpolation}), we conclude $\|\jap{A}^{\nu} R \jap{A}^{-\nu}\| \leq \|R\|^{1-\nu} \|\jap{A}R\jap{A}^{-1}\|^{\nu} <\infty$.
\end{proof}

\begin{lem}[Interpolation, {\cite[Proposition~6.17]{amrein2009}}] \label{lem:Interpolation}
	Let $X$ be a bounded operator and $S_1,S_2$ positive invertible self-adjoint operators. Assume that $S_1$ or $S_2$ is bounded. If the closure of $S_1XS_2$ is bounded, then, for $\nu\in[0,1]$, the closure of $S_1^\nu X S_2^{\nu}$ is also bounded and
	\begin{align}
		\|S_1^\nu X S_2^{\nu}\| \leq \|X\|^{1-\nu} \|S_1 X S_2\|^\nu.
	\end{align}
\end{lem}

\subsection{Locally smooth operators} \label{ssec:LocallySmoothOperators}

We introduce locally smooth operators for a family $H=(H_1,\dots,H_n)$ of strongly commuting self-adjoint operators. Most of the results of this section are straightforward generalisations of those for a single self-adjoint operator (see e.g.~\cite[Section~7.1]{amrein1996}). The only exception is Kato's Theorem (Theorem~\ref{thm:KatoSmoothness}), whose proof requires a new idea. We denote by $E$ the spectral measure of the family $H$ and by $\sigma(H)\subset\mathbb{R}^n$ the joint spectrum. Note that the intersection $D(H)=D(H_1)\cap\dots\cap D(H_n)$ is dense in $\mathcal{H}$. We consider $D(H)$ as a Banach space equipped with the graph topology.

\begin{figure}
	\begin{center}
		\begin{tikzpicture}
			
			\draw[->] (-0.5,0) -- (7,0) node[right] {$x_1$};
			\draw[->] (0,-0.5) -- (0,7) node[above] {$x_2$};
			
			\draw[dotted] (-0.5,2.5) -- (7,2.5);
			\draw[dotted] (-0.5,4.5) -- (7,4.5);
			
			\draw[dotted] (2.5,-0.5) -- (2.5,7);
			\draw[dotted] (4.5,-0.5) -- (4.5,7);
			
			\draw (2.5,2.5) rectangle (4.5,4.5) node[pos=0.5] {$K$};		
			
			\node at (3.5,1.25) {$K(\{1\})$};
			\node at (3.5,5.75) {$K(\{1\})$};
			\node at (1.25,3.5) {$K(\{2\})$};
			\node at (5.75,3.5) {$K(\{2\})$};
			\node at (1.25,1.25) {$K(\emptyset)$};	
			\node at (5.75,1.25) {$K(\emptyset)$};
			\node at (1.25,5.75) {$K(\emptyset)$};
			\node at (5.75,5.75) {$K(\emptyset)$};
			
		\end{tikzpicture}
	\end{center}
	\caption{Partition of the set $\mathbb{R}^2$ into the sets $K(a)$, $a\subset\{1,2\}$.}
	\label{fig:Partition}
\end{figure}
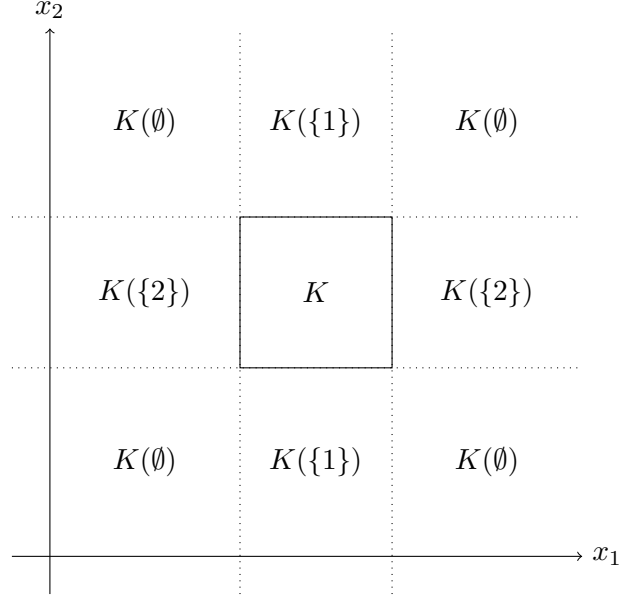

Set $\mathcal{N}=\{1,\dots,n\}$ and let $a=\{a_1,\dots,a_k\}\subset\mathcal{N}$ be a subset (w.l.o.g.~$a_1<\dots<a_k$) with $|a|=k$ elements. For $x\in \mathbb{R}^n$, we denote the vector $(x_{a_1},\dots,x_{a_k})\in\mathbb{R}^k$ by $x_a$, and the vector $(x_{b_1},\dots,x_{b_{n-k}})$ by $x^a$, where $\{b_1,\dots,b_{n-k}\}=\mathcal{N}\backslash a$ and $b_1<\dots<b_{n-k}$. We identify $x$ with $x_a\oplus x^a$. For a subset $K\subset \mathbb{R}^n$, we define the following sets (see Figure~\ref{fig:Partition}):
\begin{align}
	K(\mathcal{N}) &= K, \\
	K(a) &= \{x\in\mathbb{R}^n\mid \exists y \in \mathbb{R}^{n-|a|} \colon x_a \oplus y \in K\}, \ \emptyset\neq a\subsetneq\mathcal{N}, \\
	K(\emptyset) &= \mathbb{R}^n\backslash \bigcup_{\emptyset\neq a \subset\mathcal{N}} K(a).
\end{align}
Observe that the sets $K(a)$, $a\subset\mathcal{N}$, cover $\mathbb{R}^n$, $K$ is contained in $K(a)$ if $a\neq \emptyset$, and  
\begin{align} \label{eq:CoveringComplement}
	\mathbb{R}^n\backslash K = K(\emptyset) \cup \bigcup_{\emptyset\neq a\subsetneq\mathcal{N}} (K(a)\backslash K).
\end{align}
For products of resolvents of $H_1,\dots,H_n$, we use the following notation ($\lambda,\mu\in\mathbb{R}^n$):
\begin{align}
	R_a(\lambda+\I\mu) &= \prod_{j\in a} (H_j-\lambda_j-\I\mu_j)^{-1}, \\
	\Im R_a(\lambda+\I\mu) &= \prod_{j\in a} \Im(H_j-\lambda_j-\I\mu_j)^{-1} \notag \\
	&=\mu_{a_1}\dots\mu_{a_k} R_a(\lambda\pm\I\mu)^* R_a(\lambda\pm\I\mu). \label{eq:ImaginaryPartResolvent}
\end{align}
We abbreviate $R_\mathcal{N}(\lambda+\I\mu)$ by $R(\lambda+\I\mu)$. The following definition is a natural generalisation of locally $H$-smooth operators for a family of commuting self-adjoint operators. 

\begin{defn}
	Let $\mathcal{G}$ be a Hilbert space. A continuous operator $T:D(H)\to\mathcal{G}$ is \textbf{locally $H$-smooth} on an open set $J\subset\mathbb{R}^n$ if, for every $\emptyset\neq a\subset \mathcal{N}$ and every compact subset $K\subset J$, a constant $C_{K(a)}$ exists such that, for all $f\in\mathcal{H}$,
	\begin{align}\label{eq:HsmoothOperator}
		\int_{\mathbb{R}^{|a|}} \|T\e^{\I x_a\cdot H_a}E(K(a))f\|_{\mathcal{G}}^2 \diff x_a \leq C_{K(a)} \|f\|_{\mathcal{H}}^2.
	\end{align}
\end{defn}

In the case of a single self-adjoint operator $H$ (i.e.~$n=1$), a continuous operator $T:D(H)\to\mathcal{G}$ is locally $H$-smooth on $J\subset\mathbb{R}$ if, for every compact subset $K\subset J$,
\begin{align}
	\int_{\mathbb{R}} \|T\e^{\I x H} E(K) f\|_\mathcal{G}^2 \diff x \leq C_K \|f\|_{\mathcal{H}}^2.
\end{align}
This definition coincides with the one provided in \cite[p.274]{amrein1996}. It is possible to define locally $H$-smooth operators by demanding \eqref{eq:HsmoothOperator} only for $a=\mathcal{N}$, and some of the subsequent results can be generalised if this weaker definition is used instead. However, we prefer the above definition because Theorem~\ref{thm:KatoSmoothness} offers an equivalent characterisation of locally $H$-smooth operators in terms of resolvent estimates.

We state two identities that will be useful below. Define $\mathbb{R}_+=(0,\infty)$. For every $\mu\in \mathbb{R}^n_+$,
\begin{align}\label{eq:ResolventIdentity}
	R_a(\lambda+\I\mu) = \I^{|a|} \int_{\mathbb{R}^{|a|}_+} \e^{\I \lambda_a\cdot x_a} \e^{-\I x_a\cdot H_a-\mu_a\cdot x_a} \diff x_a.
\end{align}
The second identity is a consequence of the resolvent identity \eqref{eq:ResolventIdentity} (see \cite[(7.1.2), (7.1.11)]{amrein1996}):
\begin{align}\label{eq:ResolventIdentity2}
	\int_{\mathbb{R}^{|a|}} \| T\e^{\I x_a\cdot H_a}f\|^2 \diff x_a = \frac{2^{|a|}}{\pi^{|a|}} \sup_{\mu\in(0,1)^n} \int_{\mathbb{R}^{|a|}} \|T\Im R_a(\lambda+\I\mu)f\|^2 \diff \lambda_a.
\end{align}	

\begin{prop}\label{prop:OptimalConstant}
	If $T$ is locally $H$-smooth, the optimal constant for the bound \eqref{eq:HsmoothOperator} is 
	\begin{align}
		C_{K(a)}^0 = 2^{|a|} \sup_{\lambda\in\mathbb{R}^n,\mu\in(0,1)^n} \mu_{a_1}\dots\mu_{a_k} \|TE(K(a))R_a(\lambda+\I\mu)\|^2.
	\end{align}
	Moreover, a continuous operator $T:D(H)\to\mathcal{G}$ is locally $H$-smooth on $J$ if $C_{K(a)}^0<\infty$ for every $\emptyset \neq a\subset\mathcal{N}$ and every compact subset $K\subset J$.
\end{prop}

\begin{proof}
	The proof is similar to Step~(i) of the proof of \cite[Proposition~7.1.1]{amrein1996}. For $\mu\in(0,1)^n$, we obtain the following estimate from the resolvent identity \eqref{eq:ResolventIdentity} and the Cauchy--Schwarz inequality:
	\begin{align}
		\|TE(K(a))R_a(\lambda+\I\mu)f\|^2 &\leq \left( \int_{\mathbb{R}^{|a|}_+} \e^{-\mu_a \cdot x_a} \|T\e^{-\I x_a\cdot H_a}E(K(a))f\| \diff x_a \right)^2 \notag \\
		&\leq \frac{1}{2^{|a|} \mu_{a_1}\dots\mu_{a_k}} \int_{\mathbb{R}^{|a|}} \|T\e^{-\I x_a\cdot H_a}E(K(a))f\|^2 \diff x_a.
	\end{align}
	Applying the assumption that $T$ is locally $H$-smooth on $J$, yields $C_{K(a)}^0\leq C_{K(a)}$ (i.e.~$C_{K(a)}^0$ is the optimal constant for \eqref{eq:HsmoothOperator}). 
	
	Next, we establish that $T$ is locally $H$-smooth on $J$ if $C_{K(a)}^0<\infty$ for every $\emptyset\neq a\subset\mathcal{N}$ and every compact subset $K\subset J$. From \eqref{eq:ResolventIdentity2} and \eqref{eq:ImaginaryPartResolvent}, it follows that
	\begin{align}
		&\int_{\mathbb{R}^{|a|}} \| T\e^{\I x_a\cdot H_a} E(K(a))f\|^2 \diff x_a 
		= \frac{2^{|a|}}{\pi^{|a|}} \sup_{\mu\in(0,1)^n} \int_{\mathbb{R}^{|a|}} \|T\Im R_a(\lambda+\I\mu)E(K(a))f\|^2 \diff \lambda_a \notag \\
		&\leq \frac{2^{|a|}}{\pi^{|a|}} \sup_{\lambda\in\mathbb{R}^n,\mu\in(0,1)^n} \mu_{a_1}^2\dots\mu_{a_k}^2 \|TE(K(a))R_a(\lambda+\I\mu)\|^2 \int_{\mathbb{R}^{|a|}} \|R_a(\lambda-\I\mu)f\|^2 \diff \lambda_a.
	\end{align}
	To conclude, we utilise the following identity:
	\begin{align}
		\mu_{a_1}\dots\mu_{a_k} \int_{\mathbb{R}^{|a|}} \|R_a(\lambda-\I\mu)f\|^2 \diff \lambda_a = \pi^{|a|}\|f\|^2.
	\end{align}
	Thus, $T$ is locally $H$-smooth on $J$ if $C_{K(a)}^0<\infty$. 
\end{proof}

Proposition~\ref{prop:OptimalConstant} demonstrates the close connection between the notion of locally $H$-smooth operators and the boundary values of the resolvents $R_a(\lambda+\I\mu)$ as $\mu\downarrow0$. This connection is further clarified in Kato's theorem, which we present and prove in a generalised form below. To prepare the theorem, we repeat the above proposition in the case $n=1$.

\begin{prop}\label{prop:OptimalConstant1}
	If $n=1$, a continuous operator $T:D(H)\to\mathcal{G}$ is locally $H$-smooth if and only if $C_K^0<\infty$ for every compact subset $K\subset J$. Moreover,
	\begin{align}\label{eq:OptimalConstant1}
		C_K^0 \leq 8 \sup_{\lambda\in K,\mu\in(0,1)} \|T\Im R(\lambda+\I\mu)T^*\|.
	\end{align}
\end{prop}

The proof of \eqref{eq:OptimalConstant1} is Step~(ii) in the proof of \cite[Proposition~7.1.1]{amrein1996}.

\begin{thm} \label{thm:KatoSmoothness}
	A continuous operator $T:D(H)\to\mathcal{H}$ is locally $H$-smooth on $J$ if and only if, for every $\emptyset\neq a\subset\mathcal{N}$ and every compact subset $K\subset J$,
	\begin{align}\label{eq:Kato}
		\sup_{\lambda\in K, \mu\in(0,1)^n} \| T\Im R_a(\lambda+\I\mu)T^*\| < \infty.
	\end{align}
\end{thm}

\begin{proof}
	The strategy of the proof is to demonstrate that \eqref{eq:Kato} is equivalent to $C_{K(a)}^0<\infty$. If this is proved, the theorem follows from Proposition~\ref{prop:OptimalConstant}. We observe that, for both directions of the proof, it suffices to consider compact hyperrectangles in $J$ instead of arbitrary compact subsets $K\subset J$. This follows from the fact that every compact set in $J$ can be covered by finitely many compact hyperrectangles in $J$.
	
	Let $T$ be locally $H$-smooth on $J$. We show \eqref{eq:Kato} for $a=\mathcal{N}$. The general case is similar. Let $K=I_1\times\dots\times I_n \subset J$ be a compact hyperrectangle, where $I_1,\dots,I_n\subset\mathbb{R}$ are compact intervals. Let $\tilde{K}=\tilde{I}_1\times\dots\times\tilde{I_n} \subset J$ be another compact hyperrectangle such that, for every $j\in\mathcal{N}$, $I_j\subset \tilde{I}_j$ and $\mathrm{dist}(I_j,\mathbb{R}\backslash\tilde{I}_j)>0$. Let $\lambda\in K$. From \eqref{eq:CoveringComplement} applied to $\tilde{K}$, it follows that 
	\begin{align}
		E(\mathbb{R}^n\backslash \tilde{K}) \leq E(\tilde{K}(\emptyset)) + \sum_{\emptyset\neq a \subsetneq \mathcal{N}} E(\tilde{K}(a)\backslash \tilde{K}),
	\end{align}
	and, accordingly, 
	\begin{align}
		&\|TE(\mathbb{R}^n\backslash \tilde{K})R(\lambda+\I\mu)\|^2 \leq \|TE(\tilde{K}(\emptyset))R(\lambda+\I\mu)\|^2 + \sum_{\emptyset\neq a\subsetneq \mathcal{N}} \|TE(\tilde{K}(a)\backslash \tilde{K})R(\lambda+\I\mu)\|^2 \notag \\
		&\leq \|TE(\tilde{K}(\emptyset))R(\lambda+\I\mu)\|^2 + \sum_{\emptyset\neq a\subsetneq \mathcal{N}} \|TE(\tilde{K}(a)\backslash \tilde{K})R_a(\lambda+\I\mu)\|^2 \|E(\tilde{K}(a)\backslash \tilde{K})R_{\mathcal{N}\backslash a}(\lambda+\I\mu)\|^2;
	\end{align}
	hence, by Proposition~\ref{prop:OptimalConstant},
	\begin{align}
		&\sup_{\lambda\in K,\mu\in(0,1)^n} \mu_1\dots\mu_n \|TE(\mathbb{R}^n\backslash\tilde{K})R(\lambda+\I\mu)\|^2 \notag \\
		&\leq \sup_{\lambda\in K} \|TE(\tilde{K}(\emptyset))R(\lambda)\|^2 + \sum_{\emptyset\neq a\subsetneq \mathcal{N}} \frac{1}{2^{|a|}} C_{\tilde{K}(a)}^0 \sup_{\lambda\in K} \|E(\tilde{K}(a)\backslash \tilde{K}) R_{\mathcal{N}\backslash a}(\lambda)\|^2 < \infty,
	\end{align}
	where we used that $C_{\tilde{K}(a)}^0<\infty$ due to the assumption that $T$ is locally $H$-smooth on $J$.
	Thus, we obtain \eqref{eq:Kato}:
	\begin{align}
		\sup_{\lambda\in K, \mu\in(0,1)^n} \| T\Im R(\lambda+\I\mu)T^*\| 
		\leq \sup_{\lambda\in K, \mu\in(0,1)^n} \mu_1\dots\mu_n \|TE(\tilde{K})R(\lambda+\I\mu)\|^2 \notag \\ 
		+ \sup_{\lambda\in K, \mu\in(0,1)^n} \mu_1\dots\mu_n \|TE(\mathbb{R}^n\backslash\tilde{K})R(\lambda+\I\mu)\|^2 < \infty.
	\end{align}
	Next, we prove that $T$ is locally $H$-smooth on $J$ if \eqref{eq:Kato} holds for every $\emptyset \neq a \subset \mathcal{N}$ and every compact subset $K\subset J$. We demonstrate that $C_{K}^0<\infty$ for every compact hyperrectangle $K=I_1\times\dots\times I_n$. The case $C_{K(a)}^0<\infty$ is similar. We consider the contributions in the supremum defining $C_K^0$ from the points $\lambda\notin K$. If $\lambda\notin K$, an $a\subsetneq\mathcal{N}$ exists such that $\lambda \in K(a)$. For every $j\in \mathcal{N}\backslash a$, we choose $\kappa_j\in I_j$ such that $\mathrm{dist}(\lambda_j,I_j)=|\lambda_j-\kappa_j|$, and we define $\tilde{\lambda}\in K$ to be the element that satisfies $\tilde{\lambda}_j=\lambda_j$ if $j\in a$ and $\tilde{\lambda}_j=\kappa_j$ if $j\in \mathcal{N}\backslash a$. We have
	\begin{align}
		R(\lambda+\I\mu) &= \prod_{j\in\mathcal{N}\backslash a}(1+(\lambda_j-\kappa_j)(H_j-\lambda_j-\I\mu_j)^{-1}) R(\tilde{\lambda}+\I\mu);
	\end{align}
	thus,
	\begin{align}
		\|TE(K)R(\lambda+\I\mu)\| &\leq \|TE(K)R(\tilde{\lambda}+\I\mu)\| \|E(K)\prod_{j\in\mathcal{N}\backslash a} (1+(\lambda_j-\kappa_j)(H_j-\lambda_j-\I\mu_j)^{-1})\| \notag \\
		&\leq 2^{n-|a|} \|TE(K)R(\tilde{\lambda}+\I\mu)\|.
	\end{align}
	It follows that
	\begin{align} \label{eq:OptimalConstantN}
		C_K^0 \leq 8^n \sup_{\lambda\in K,\mu\in(0,1)^n} \|T\Im R(\lambda+\I\mu)T^*\| < \infty.
	\end{align}
	This estimate can be compared with \eqref{eq:OptimalConstant1}. However, we have obtained the factor $8^n$ only in the case that $K\subset J$ is a compact hyperrectangle, and \eqref{eq:OptimalConstantN} might not generalise to arbitrary compact sets $K$.
\end{proof}

In the remainder of this section, we discuss an example of a locally smooth operator, which is important for the main part of the paper. Let $\mathcal{H}_\mathrm{ac}(H)\subset\mathcal{H}$ be the jointly absolutely continuous subspace of $H$. 

\begin{defn}\label{defn:RadonNikodymDerivativeBounded}
	For $f\in \mathcal{H}_{\mathrm{ac}}(H)$, let $\rho_f$ be the Radon--Nikodym derivative (w.r.t.~the Lebesgue measure on $\mathbb{R}^n$) of the spectral measure $\scp{f}{E(\cdot)f}$. For $a\subset\mathcal{N}$, we denote by $\mathcal{M}(H)_a$ the set of all vectors $f\in \mathcal{H}_{\mathrm{ac}}(H)$ for which
	\begin{align}
		\normm{f}_a = \sup_{x_a\in\mathbb{R}^{|a|}} \left(\int_{\mathbb{R}^{n-|a|}} \rho_f(x_a\oplus x^a) \diff x^a\right)^\frac{1}{2} < \infty.
	\end{align}
	We abbreviate $\mathcal{M}(H)_\mathcal{N}$ by $\mathcal{M}(H)$ and $\normm{f}_{\mathcal{N}}$ by $\normm{f}$.
\end{defn} 

Observe that $\mathcal{M}(H)_\emptyset = \mathcal{H}_{\mathrm{ac}}(H)$ and that $\mathcal{M}(H)$ is the set of all vectors $f\in \mathcal{H}_{\mathrm{ac}}(H)$ for which the Radon--Nikodym derivative $\rho_f$ is a bounded function. In general, for $\emptyset \neq a \subsetneq \mathcal{N}$, $\mathcal{M}(H)$ is not a subset of $\mathcal{M}(H)_a$. If, for every $j\in\mathcal{N}\backslash a$, an $h_j\in\mathbb{R}_+$ exists such that $f=E(\{|x_j|\leq h_j\})f$, then
\begin{align}\label{eq:BoundedEnergyEstimate}
	\normm{f}_a^2 = \sup_{x_a\in\mathbb{R}^{|a|}} \int_{\mathbb{R}^{n-|a|}} \rho_f(x_a\oplus x^a) \diff x^a 
	\leq \prod_{j\in\mathcal{N}\backslash a} 2h_j \normm{f}^2.
\end{align}
This follows from the fact that in this case $\rho_f(x) = \prod_{j\in\mathcal{N}\backslash a} \chi(|x_j|\leq h_j) \rho_f(x)$, where $\chi$ denotes the characteristic function.

\begin{lem} \label{lem:MSpaceDense}
	For every $a\subset\mathcal{N}$, $\mathcal{M}(H)_a$ is dense in $\mathcal{H}_{\mathrm{ac}}(H)$ (w.r.t.~the norm in $\mathcal{H}$).
\end{lem}

\begin{proof}
	It suffices to prove that $\mathcal{M}(H)$ is dense in $\mathcal{H}_{\mathrm{ac}}(H)$ because $f\in\mathcal{M}(H)$ is approximated by $f_R = E(\{|x|\leq R\})f$ as $R\to\infty$ and $f_R \in \mathcal{M}(H)_a$ due to \eqref{eq:BoundedEnergyEstimate}. We may assume that a cyclic vector $f\in\mathcal{H}_{\mathrm{ac}}(H)$ for $H$ exists. Otherwise, we decompose $\mathcal{H}_{\mathrm{ac}}(H)$ into subspaces such that each subspace has a cyclic vector. By spectral theory, the Hilbert space $\mathcal{H}_{\mathrm{ac}}(H)$ is unitarily equivalent to $L^2(\sigma(H),\rho_f\diff x)$ in such a way that $H$ is unitarily equivalent to the multiplication operator by $x$. For $g \in L^2(\sigma(H),\rho_f\diff x)$, we have $\rho_g = |g|^2\rho_f$. Because every $g\in L^2(\sigma(H),\rho_f\diff x)$ is approximated by $g\chi(|g|^2\rho_f\leq R)$ as $R\to\infty$, the space $\mathcal{M}(H)$ is dense in $\mathcal{H}_{\mathrm{ac}}(H)$.
\end{proof}

The following proposition demonstrates that the operator $T_f:\mathcal{H}\to \mathbb{C}$, $g\mapsto \scp{f}{g}$, is (locally) $H$-smooth on $\mathbb{R}^n$ if $f\in\bigcap_{a\subset\mathcal{N}}\mathcal{M}(H)_a$. The proposition is a generalisation of \cite[Lemma XI.3.1]{reed1979}.

\begin{prop} \label{prop:SquareIntegrability}
	For $\emptyset \neq a\subset\mathcal{N}$, $f \in \mathcal{M}(H)_a$ and $g\in\mathcal{H}$,
	\begin{align}
		\int_{\mathbb{R}^{|a|}} |\scp{f}{\e^{\I x_a\cdot H_a}g}|^2 \diff x_a \leq (2\pi)^{|a|} \normm{f}^2_a \norm{g}^2.
	\end{align}
	If, for every $j\in\mathcal{N}\backslash a$, an $h_j\in\mathbb{R}_+$ exists such that $f=E(\{|x_j|\leq h_j\})f$, then
	\begin{align} \label{eq:TfHSmoothBoundedEnergy}
		\int_{\mathbb{R}^{|a|}} |\scp{f}{\e^{\I x_a\cdot H_a}g}|^2 \diff x_a \leq (2\pi)^{|a|} \prod_{j\in\mathcal{N}\backslash a} 2h_j \normm{f}^2 \norm{g}^2.
	\end{align}
\end{prop}

\begin{proof}	
	We demonstrate that $T_f$ is locally $H$-smooth on $\mathbb{R}^n$. Let $K\subset \mathbb{R}^n$ be a compact subset. By Proposition~\ref{prop:OptimalConstant}, it suffices to show that, for every $\emptyset \neq a\subset\mathcal{N}$,
	\begin{align}
		\sup_{\lambda\in\mathbb{R}^n,\mu\in(0,1)^n} \mu_1\dots\mu_n \|T_f E(K(a)) R_a(\lambda+\I\mu)\|^2 
		&\leq \sup_{\lambda\in\mathbb{R}^n,\mu\in(0,1)^n} \|T_f \Im R_a(\lambda+\I\mu) T_f^* \| \notag \\
		&\leq \pi^{|a|} \normm{f}^2_a.
	\end{align}
	The first inequality is obvious. To obtain the second inequality, we observe the following:
	\begin{align}
		&\sup_{\mu\in(0,1)^n} \|T_f \Im R_a(\lambda+\I\mu) T_f^* \| 
		= \sup_{\mu\in(0,1)^n} |\scp{f}{\Im R_a(\lambda+\I\mu)f}|
		\notag \\
		&=\sup_{\mu\in(0,1)^n} \int_{\mathbb{R}^n} \prod_{j\in a} \Im \frac{1}{\sigma_j-\lambda_j-\I\mu_j} \rho_f(\sigma) \diff \sigma
		\leq \pi^{|a|} \int_{\mathbb{R}^{n-|a|}} \rho_f(\lambda_a\oplus \sigma^a) \diff\sigma^a.
	\end{align}
	We conclude that, for every compact subset $K\subset\mathbb{R}^n$,
	\begin{align}
		\int_{\mathbb{R}^{|a|}} |T_f \e^{\I x_a\cdot H_a}E(K(a))g|^2 \diff x_a \leq (2\pi)^{|a|} \normm{f}^2_a \|g\|^2.
	\end{align}
	The first statement of the lemma follows from Fatou's lemma because the bound on the r.h.s.~is independent of $K$. The second statement follows from \eqref{eq:BoundedEnergyEstimate}.
\end{proof}

\end{document}